\documentclass[sn-mathphys,Numbered]{sn-jnl}
\usepackage{amsmath,amssymb,amsfonts}
\usepackage{geometry}
\usepackage{graphicx, epstopdf}
\usepackage{tensor}
\usepackage{hyperref}
\usepackage{setspace}
\usepackage{ragged2e} 
\usepackage{amsmath,amssymb,amsthm}
\usepackage{placeins}
\usepackage{url}  
\usepackage[numbers]{natbib}
\usepackage[utf8]{inputenc}
\usepackage{graphicx}%
\usepackage{multirow}%
\usepackage{xcolor}
\usepackage{amsthm}%
\usepackage{mathrsfs}%
\usepackage[title]{appendix}%
\usepackage{xcolor}%
\usepackage{textcomp}%
\usepackage{manyfoot}%
\usepackage{booktabs}%
\usepackage{algorithm}%
\usepackage{algorithmicx}%
\usepackage{algpseudocode}%
\usepackage{listings}%
\usepackage{csquotes}
\usepackage{microtype}

\newtheorem{theorem}{Theorem}
\newtheorem{proposition}{Proposition}
\newtheorem{corollary}{Corollary}
\theoremstyle{definition}
\newtheorem{definition}{Definition}

\numberwithin{equation}{section}

\date{}

\begin{document}
	\title[Article Title]{Traversable Wormhole De-singularization: Almost $\eta$-Ricci-Yamabe Solitons in Static Spherically Symmetric Imperfect Fluid Spacetimes}
	\author[1]{\fnm{Jay Prakash} \sur{Singh}}\email{jpsmaths@cusb.ac.in}	
	\equalcont{These authors contributed equally to this work.}
	\author[2]{\fnm{Jaswant}}\email{jaswantj385@gmail.com}
	\equalcont{These authors contributed equally to this work.}
	\affil[1,2]{\orgdiv{Department of Mathematics}, \orgname{Central University of South Bihar}, \orgaddress{ \city{Gaya}, \postcode{824236}, \state{Bihar}, \country{India}}}
	
	\justifying
	\abstract{In this paper, we investigate the almost $\eta$-Ricci-Yamabe soliton as a fundamental geometric regulator for a static, spherically symmetric black hole coupled to an imperfect fluid. We have shown that the scaling parameter $\omega(r)$ is governed by thermodynamic friction along the radial vector field, and the geometric coupling with the Hawking temperature: $\alpha(r_H) S_{tt} = 2\pi T_H$, at the horizon. We also derive the Poisson equation along the gradient vector field of the soliton and prove that the flow's kinematic expansion is explicitly dependent on the fluid's equation of state $\rho = \gamma \sigma$. Diverging from traditional methodologies that assume a geometric shape function apriori, we analytically proved the geometric flow endogenously transitions the black hole geometry into a traversable wormhole throat, by regularizing of temporal coordinate and satisfying spatial flare-out condition. This transition occurs when fluid enters the dark energy era at $\gamma = -1$, and violates the Null Energy Condition $\rho + \sigma < 0$, with the soliton strictly dominating the local curvature gradient $\omega^{\prime}(r_H) > f^{\prime\prime}(r_H)$, to keep the throat open. Moreover, by smoothly attenuating at spatial infinity, the soliton preserves the exact cosmological spacetime. Finally, through tensorial perturbation analysis, we demonstrate that the geometric flow introduces a localized dissipative mechanism, that the perturbation evolution reduces to damped wave equation, imposing geometric drag on the manifold.}  
	
	\keywords{ Traversable wormholes, Black hole, Almost $\eta$-Ricci-Yamabe solitons, Imperfect fluid, static spherically symmetric spacetime, Reissner-Nordström-de Sitter, Null Energy Condition, Perturbation.}
	\maketitle
	\section{Introduction}
	Recently, the study of geometric flow and solitons has evolved into a profound framework in differential geometry and mathematical physics. The wormholes are a very fascinating topic in mathematical physics over the last couple of decades, theoretically derived from Einstein’s field equations, they are hypothetical spacetime tunnels connecting distinct regions of the universe \cite{hochberg1997geometric,visser1997geometric}. Their defining geometric feature is the wormhole throat, which is the narrowest pinch of the tunnel, where the radius reaches its absolute minimum. Also, when the matter or light can safely pass through the tunnel, they are known as traversable wormholes \cite{nath2024new}. A traversable wormhole static with a throat and two mouths open constructed by Morris \& Throne \cite{morris1988wormholes}. According to framework of General Relativity \cite{einstein1915feldgleichungen} a region of spacetime where extreme gravitational collapse and an immense amount of matter are packed in very small area. It prevents any physical quantity including light, from escaping is known as a black hole and the boundary, where gravity is so strong that it sucks everything inside the blcak hole is called the horizon \cite{hawking2023large,wald2010general}. Therefore, traversable wormholes characterized as absence of horizon because horizon can prevent two-way travel \cite{morris1988wormholes}.\\
	Now, we shall begin with the required mathematical formulation. Firstly, the Ricci flow introduced by Hamilton \cite{hamilton1982three} in 1982. Later, study by many authors see \cite{cao2009recentprogressriccisolitons, cao2009geometrycompletegradientshrinking, perelman2003ricci} it is defined as:
	\[ \frac{\partial }{\partial t}g(t) = -2 S(t), \hspace{1cm}t \ge 0, \hspace{0.5cm} g(0) = g, \]
	where $g$ and $S$ are denotes the Riemannian metric and the $(0,2)$ - symmetric Ricci tensor respectively. Consider a Riemannian manifold $(M^n,g)$, confirms a Ricci soliton which is simply generalized an Einstein metric corresponding to self similar soliton. Also, if there exits smooth vector field denoted by $V$, and a constant denoted by $\lambda$, on the manifold given as: \[  \frac{1}{2} \mathcal{L}_{V}g + S =\lambda g,\] 
	where the Lie derivative denoted by $ \mathcal{L}_{V} g $ along the vector field $ V $ \cite{cao1996existence, cao2009recentprogressriccisolitons}. In 1960, Yamabe introduced Yamabe's problem \cite{yamabe1960deformation}, after that Yamabe flow was introduced by Hamilton \cite{hamilton1989lectures} which is conformal to Riemannian metric $g$, which is defined as follows:
	\[ \frac{\partial }{\partial t}g(t) = -\mathcal{R} g(t), \]
	where $ \mathcal{R}$ denodes the scalar curvature \cite{khatri2022study}. Just like Ricci soliton, it is also a self-similar soliton it is defined as: \[  \frac{1}{2} \mathcal{L}_{V}g  =(\lambda - \mathcal{R} )g. \]
	The mathematical significance of these  solitons was acknowledged. When G. Perelman use gradient Ricci solitons to provide spectacular proof of the Poincaré conjecture \cite{perelman2002entropy, perelman2003ricci}. A method further formalized in literature to mathematically synthesize the advantage of both deformations. Therefore, the combined Ricci-Yamabe or $(\alpha,\beta)$-Ricci-Yamabe flow was introduced, which is defined as follows:
	\begin{definition} \cite{guler2019ricci}
		Let a map $RY^{(\alpha, \beta, \omega, g)} : I \to T_2^s(M)$, given as:
		
		\[RY^{(\alpha, \beta, \omega, g)} = \frac{\partial g}{\partial t} + 2\alpha S(t) + \beta \mathcal{R}g(t) \tag{I} \label{I},\]
		
		then equation (\ref{I}) is said to be Ricci-Yamabe or $(\alpha,\beta)$-Ricci-Yamabe map of the given manifold.\\ 
		If $$RY^{(\alpha, \beta, \omega, g)} \equiv 0,$$ then the $g(.)$ is said to be Ricci-Yamabe or $(\alpha,\beta)$-Ricci-Yamabe flow. 
	\end{definition}
	Depending on the signs of $\alpha$ and $\beta$, the Ricci-Yamabe flow can be singular, pseudo-Riemannian, or Riemannian, providing flexibility and helping the study of spacetime geometry or the analysis of physical models of relativistic theories. Now according to Dey \cite{dey2020almost} the Ricci-Yamabe or $(\alpha,\beta)$-Ricci-Yamabe  soliton  is defined as: 
	\begin{definition} \cite{dey2020almost}
		A Riemannian or pseudo-Riemannian manifold $(M^n, g)$ is said to be Ricci-Yamabe soliton $(g, V, \mu, \alpha, \beta, \omega)$ if
		\[  \frac{1}{2}\mathcal{L}_{V}g + \alpha S  = \left( \lambda - \frac{\beta \mathcal{R}}{2}\right)g.\, \]
		
		So the Ricci-Yamabe or $(\alpha,\beta)$-Ricci-Yamabe soliton is locally steady, shrinking or expanding depends on whether $\lambda = 0$, $\lambda <  0$ or $\lambda > 0$ respectively. 
	\end{definition}
	The gradient Ricci-Yamabe soliton characterized by the existence of a smooth function $ f: M \to \mathbb{R} $ such that $ V = Df $, where $ D$ is denotes the gradient operator corresponding to the metric \( g \) on the manifold. Also, the Ricci-Yamabe soliton is simply generalization of the Ricci and Yamabe solitons, which is well known Einstein soliton \cite{catino2016gradient, dwivedi2021some, venkatesha2019gradient}. However, in physical spacetime, geometry does not evolve in vacuum, it is coupled with covariant vector fields generated by matter distributions or spacetime symmetries \cite{blaga2014eta, cho2009ricci}. Also, some results ``On Ricci-Yamabe soliton and geometrical structure in perfect fluid spacetime'' explained by Singh \& Khatri \cite{singh2021ricci}.
	In 2009 Ricci Yamabe soliton to $\eta$-Ricci Yamabe soliton extended by Cho \& Kimura \cite{cho2009ricci}, after that more generalization of $\eta$-Ricci-Yamabe soliton type $(\alpha,\beta)$ given by Siddiqui and Akyol \cite{siddiqi2020eta}, defined as:
	\begin{equation}
		\frac{1}{2} \mathcal{L}_{V}g + \alpha S + \left(\lambda - \frac{\beta \mathcal{R}}{2}\right) +  \omega \eta \otimes \eta =0,
	\end{equation}
	with that soliton is steady when $\lambda = 0$, shrinking when $\lambda < 0$, and expanding when $\lambda > 0$. Here $\omega$ denotes the scaling factor, $\eta$ denotes the associated 1-form, $\alpha$ and $\beta$ denotes the Ricci and Yamabe perturbation parameter respectively \cite{siddiqi2020eta}. When $\omega = 0 $, then it reduces to the Ricci-Yamabe soliton. If the constants $\alpha$, $\beta$, $\lambda$ and $\omega$ are treated as smooth radial functions ($\alpha(r)$, $\beta(r)$, $\lambda(r)$ and $\omega(r)$) on the manifold to dynamically couple the geometric flow, that radially varying thermodynamic pressure of the surrounding imperfect fluid. Then the spacetime dynamically evolves under the framework of an almost $\eta$-Ricci-Yamabe soliton \cite{shaikh2026symmetry,siddiqi2020eta}.\\
	Also, we have defined the background manifold $(\mathcal{M}^4,g)$ as a 4-dimensional static, spherically symmetric spacetime, coupled to an imperfect fluid \cite{alkhaldi2021imperfect, rahman, livine2025effective, hiscock1983stability, stephani2009exact, sushkov2005wormholes}. The line element   describing this geometry is given as follows:
	\begin{equation}
		ds^2 = -f(r)dt^2 + f(r)^{-1}dr^2 + r^2(d\theta^2 + \sin^2\theta d\phi^2),  \label{1.2}
	\end{equation}
	where $f(r)$ acts as the structural lapse function with the temporal, radial, and angular coordinates defined respectively by $-\infty < t < \infty$, $r \ge 0$, and $\phi \in [0, 2\pi]$ \cite{shaikh2026symmetry}. We have also evaluated the ``Reissner-Nordström-de Sitter (RNdS) spacetime'' to verify validity of the study  and ground the topological theorems in an astrophysical framework \cite{ali2003spinning, gim2019charged, shaikh2026symmetry, stephani2009exact}. It models a spherically symmetric, electric charged black hole with cosmological constant denoted by $\Lambda$ and the structural metric function $f(r)$ for RNdS spacetime given as:
	\begin{equation}
		f(r) = 1 - \frac{2M}{r} + \frac{Q^2}{r^2} - \frac{\Lambda r^2}{3},  \label{1.3}
	\end{equation}
	where $Q$ is denotes the electrically charge, $M$ is denotes the mass of  black hole\cite{shaikh2026symmetry}. We can find an event horizon, which dictates that the global boundary of this black hole acts purely as a coordinate singularity $(f(r)=0)$, rather than a true physical divergence \cite{hawking2023large, wald2010general}. However, dynamic geometric flows are analyzed locally at the apparent horizon, which is defined as the outermost boundary where the expansion of outgoing null geodesics vanishes completely \cite{ashtekar2004isolated, penrose1965gravitational}.\\
	\\
	The thermodynamic expansion of the geometry is strictly governed by the stress-energy tensor of the surrounding matter, which is physically realized in this paper as an accretion disk. According to relativistic astrophysics, as a dynamic rotating pool of diffuse matter and plasma spiralling into the central black hole \cite{destounis2019dynamical, hawking2023large, jimenezthermodynamics, wald2010general}. Where as the classical cosmological exact solution relies on the assumption of a perfect fluid, which characterized by isotropic pressure $\rho$, energy density $\sigma$ and the energy-momentum tensor \(T\), \cite{o1983semi} defined as follows:
	\[ T_{\mu\nu} = (\rho + \sigma) u_{\mu} u_{\nu} + \rho g_{\mu\nu},\]
	since, the concept of a perfect fluid fails to describe the real astrophysical environments. Therefore recent observations have shifted towards the study of imperfect fluid. This provide a more real representation than perfect fluids and shows complex internal friction, heat flux, and atmospheric stress, \cite{abramowicz2013foundations, israel1979transient, maartens1996causal, shakura1973black}. According to barotropic equation of state \( \rho = \gamma \sigma \) where \(\gamma\) is state parameter, classifies the cosmological era \cite{capozziello2011extended, carroll2004introduction} at \(\gamma=-1 \), imperfect fluids enter the dark energy era. Under this condition fluid strictly violates the Null energy conditions (NEC), which means \( \sigma + \rho < 0 \), \cite{martin2017classical}  as pioneered by Morris \& Throne \cite{morris1988wormholes} and later generalized by Visser \cite{visser1995lorentzian}. This geometric flare-out condition is mandatory to de-singularize a black hole geometry at apparent horizon $(r=r_H) $, and sustain a traversable wormhole, which is a non-singular topological bridge that connecting two distinct regions of spacetime \cite{hochberg1997geometric, sushkov2005wormholes, visser1997geometric}. Consequently, recent mathematical physics has focused on determining whether generalized topological flows can formally map a singular trapped horizon to a regularized state \cite{husain2008ricci, lobo2017wormholes}, and whether such static geometries can theoretically sustain localized dissipative mechanisms against metric perturbations \cite{chandrasekhar1998mathematical}.\\
	To precisely contextualize the theoretical advancements of the study. We compare our framework with some recent studies in this field. Many authors contributed to generalizing the Ricci and Yamabe soliton framework and study their properties like Blaga et al. study almost $\eta$-Ricci and almost $\eta$-Yamabe solitons \cite{blaga2022almost,blaga2018almost}, Siddiqi and Akyol study $\eta$-Ricci-Yamabe soliton \cite{siddiqi2020eta}, then evolution to almost $\eta$-Ricci-Yamabe solitons see \cite{singh2021ricci, khatri2025almost, shaikh2026symmetry, mert2024pseudosymmetric,  pandey2024alpha}. The coupling parameters of the almost $\eta $ -Ricci-Yamabe soliton, can act as smooth functions rather than rigid constants \cite{shaikh2026symmetry, jafari2026generalized}. Most of the study based on perfect fluid spacetime which purely theoretical model and having no heat, friction and viscosity. Therefore recent studies move towards imperfect fluid which is more realistic model. It has friction, heat and viscosity. Alkhaldi, Siddiqi et al. study that GRW-Spacetime Admits Ricci-Yamabe Metric in imperfect fluid spacetime \cite{alkhaldi2021imperfect}. Also we can find many studies on static, spherically symmetric spacetime in perfect fluid see \cite{livine2025effective, carr2003spacetime, rahman}. In our framework we investigate the static, spherically symmetric imperfect fluid spacetime, admitting almost $\eta$-Ricci-Yamabe solitons where parameters are smooth functions and imperfect fluid have heat and friction it allows the geometry to dynamically adapt to realistic anisotropic stresses, a scenario largely unexplored in current soliton literature.\\
	Recent extensive studies by Nath and Sarma \cite{nath2024new}, Meshwa and Ahmed \cite{kurbah2024electromagnetic}, Visser et al. \cite{visser2003traversable, simpson2019black}, Konoplya and Zhidenko \cite{konoplya2022traversable}, Cataldo et al. \cite{cataldo2017traversable}, the methodology in these studies typically in a way where, authors first assume a mathematically convenient shape function $b(r)$ for the throat, and then calculate the exotic matter required to support it. In our framework, we don't assume shape function  $b(r)$ apriori, we analytically prove that when an imperfect fluid enters the dark energy era ($\gamma = -1$), the geometric flow endogenously generates the necessary effective stress and the almost $\eta$-Ricci-Yamabe flow naturally forces the the spatial flare-out condition $\omega(r_H) > 0$, and regularizes the time coordinate, to create the wormhole geometry.\\
	In recent years, authors like Mustafa et al. \cite{mustafa2021traversable}, Ilyas et al. \cite{ilyas2023traversable}, Mishra et al. \cite{mishra2022traversable} , Sahoo et al. \cite{sahoo2021traversable}, for more see \cite{ashraf2026traversable, malik2023investigation, shamir2020traversable, sokoliuk2022generalised}. They realized that exotic matter is physically problematic. To get around this they still  assume shape function $b(r)$ and change the Einstein field equations to modified gravity such as $f(\mathcal{R})$ theories. The extra mathematical terms from the modified gravity act as ``effective" exotic matter, allowing the actual physical fluid to be normal matter. Since, we don't assume shape function apriori. Therefor, integrating our framework into $f(\mathcal{R})$ modified gravity, it will simultaneously eliminate the need for both apriori shape assumptions and physical exotic matter, defines the optimal trajectory for our future research.\\
	
	\section{Notations and Preliminary}
	\justifying
	We adopt the following notations throughout this paper. ($\mu, \nu, \varrho, l, s $) letters denote spacetime indices, appearing as superscripts for contravariant and subscripts for covariant components over coordinates $(t, r, \theta, \phi)$. The background metric tensor is denoted by $g_{\mu\nu}$ with structural lapse function $f(r)$. $\Lambda$, $\kappa$,  $r_H$, $T_H$, $K$, denotes Cosmological constant, Einstein gravitational constant, the apparent horizon, Hawking temperature and surface gravity respectively. For traversable wormhole geometries, $\Phi(r)$ is the redshift function and $b(r)$ is the shape function. The Ricci curvature tensor is denoted by $S_{\mu\nu}$ and the Ricci scalar by $\mathcal{R}$. The energy-momentum tensor of imperfect fluid and electromagnetic stress-energy tensors are $T_{\mu\nu}$ and $E_{\mu\nu}$, characterized by 4-velocity $u^\mu$, trace of energy-momentum tensor is $T$ , mass-energy density $\rho$, isotropic pressure $\sigma$, equation of state parameter $\gamma$, heat flux $q_\mu$, and shear stress $\pi_{\mu\nu}$. The radial null vector is $k^\mu$. The almost $\eta$-Ricci-Yamabe soliton is generated by the radial vector field $\xi$ with dual 1-form $\eta_\mu$, utilizing smooth radial functional parameters $\alpha(r)$ and $\beta(r)$, are Ricci and Yamabe coupling, also solitonic expansion and scaling factor $\lambda(r)$, and $\omega(r)$ respectively. The solitonic scalar potential is denoted as $\Psi$, and $\Delta \Psi$ denotes the kinematic expansion. $\nabla$ is the covariant derivative, and $\Gamma_{\mu\nu}^\varrho$ denotes the affine connections (Christoffel symbols). $h_{\mu\nu}$ is the first-order metric perturbation with trace $h$, and $\Box$ represents the d'Alembertian operator.\\
	
	Let $(\mathcal{M}^4,g)$ be a 4-dimensional static, spherically symmetric spacetime manifold coupled to an imperfect fluid. Also the geometric flow defined as an almost $\eta$-Ricci-Yamabe soliton along the radial vector field $ \xi = \partial_r$.  Where the expansion parameters are smooth radial functions \cite{destounis2019dynamical,shaikh2026symmetry}. We analyze the manifold containing a charged black hole embedded in a universe with a cosmological constant $\Lambda$, where the geometry is strictly governed by the line element (\ref{1.2}).
	\begin{definition}\cite{pigola2011ricci,shaikh2026symmetry,siddiqi2020eta} \label{definition 3.} \label{def 3}
		The 4-dimensional $(\mathcal{M}^4,g)$ manifold admits an almost $\eta$-Ricci-Yamabe soliton if it satisfies along a radial vector field $\xi = \partial_r$ as: 
		\begin{equation}
			\frac{1}{2} \mathcal{L}_{\xi}g_{\mu\nu} + \alpha(r) S_{\mu\nu} + \left(\lambda(r) - \frac{\beta(r) \mathcal{R}}{2}\right)g_{\mu\nu} +  \omega(r) \eta_\mu \eta_\nu = 0,  \label{2.1}
		\end{equation}
		where $\alpha(r), \beta(r), \lambda(r)  $ and $\omega(r)$ are smooth functional parameters on manifold. Such as $\alpha(r)$ and $\beta(r)$ represent the Ricci and Yamabe parameters. $\lambda(r)$ shows that whether the geometry at distance $r$ is steady at $\lambda(r) = 0$, shrinking at $\lambda(r) < 0$, or expanding at $\lambda(r) > 0$. The solitonic scaling factor  $\omega(r)$ as a differential multiplier for the flow, and $\eta_\mu$ is the associated 1-form. The parameters are smooth functions depending exclusively on the radial coordinate $r$ as our spacetime is static and spherically symmetric. Therefore to preserve these exact symmetries, the parameters cannot depend on $t, \theta,$ or $\phi$. Physically, the flow is coupled to an imperfect fluid whose density and pressure vary only with the radial distance from the black hole. Therefore, for the geometric flow to correctly balance the fluid's physical stress, the solitonic parameters must dynamically adapt as purely radial functions.
	\end{definition}
	Before defining the energy-momentum tensor, we establish the physical motivation for using an imperfect fluid. In highly idealized models, black holes are often treated in a vacuum or surrounded by a friction-less perfect fluid. However, realistic astrophysical black holes are actively fed by surrounding accretion disks, which consist of swirling clouds of gas and plasma that naturally possess internal friction and unequal pressure distributions. Therefore, to accurately model a realistic accretion environment, we mathematically coupled the background spacetime to an imperfect fluid.
	\begin{definition} \cite{hiscock1983stability,israel1979transient,rezzolla2013relativistic} \label{definition 4.}
		The surrounding matter is characterized as a imperfect fluid ( Accretion disk) with a 4-velocity $ u^{\mu} = (1/\sqrt{f(r)}, 0, 0, 0)$ and $ u_{\mu} =(-\sqrt{f(r)}, 0, 0, 0)$ satisfying the standard relativistic normalization $ g_{\mu\nu}u^\mu u^\nu = -1 $.\\
		Then the energy-momentum tensor denoted by $T_{\mu\nu}$ without bulk viscosity, defined as: 
		\begin{equation}
			T_{\mu\nu} = (\sigma + \rho)u_{\mu }u_{\nu} + \rho g_{\mu\nu} + q_{\mu} u_{\nu} + q_{\nu} u_{\nu} + \pi_{\mu\nu}, \label{2.2}
		\end{equation}
		where $ \sigma$ and $\rho$  denotes mass-energy density and isotropic pressure respectively also $q_\mu = (0, q_r, 0, 0)$ denotes the radial heat flux, and $\pi_{\mu\nu}$ denotes the traceless anisotropic shear stress tensor, representing the fluid friction. 
	\end{definition}
	
	The evaluation of the soliton is need to know how the accretion disk and the central mass combine to warp the Ricci curvature $S_{\mu\nu}$. Then defined modified field equation where geometric matter are coupled by Einstein's equation as: 
	\begin{definition}\cite{wald2010general, carroll2004introduction, alkhaldi2021imperfect} \label{definition 5.}
		For the manifold coupled to an imperfect fluid energy-momentum tensor $T_{\mu\nu}$ and an electromagnetic field $E_{\mu\nu}$, and including the cosmological constant denots as $\Lambda$,  then the modified Einstein field equation is defined as: 
		\begin{equation}
			S_{\mu\nu} - \frac{1}{2} \mathcal{R} g_{\mu\nu} + \Lambda g_{\mu\nu} = \kappa (T_{\mu\nu} + E_{\mu\nu}), \label{2.3}
		\end{equation}
		where $\kappa = 8\pi G/c^4$ denotes the Einstein gravitational constant.
	\end{definition}
	\begin{definition}
		\cite{carroll2004introduction} \label{definition 6.}
		In classical general relativity the Null Energy Condition (NEC) state that energy-momentum tensor $T_{\mu\nu}$ satisfies $T_{\mu\nu}k^{\mu}k^{\nu} \ge 0$ for all null vector field $k^{\mu}$ or equivalently that ($\sigma + \rho \ge 0$). The intense internal friction of accretion fluid drives it into a dark energy-like vacuum state ($\rho = \gamma \sigma$, where $ \gamma \le -1$). This generates extreme negative radial pressure  actively forcing a violation of NEC ($\sigma + \rho < 0$) \cite{NEC2022criteria}.
	\end{definition}
	\begin{definition}
		\cite{adamiak2008static,morris1988wormholes, visser1995lorentzian} \label{definition 7.}
		The metric describing a static, spherically symmetric Lorentzian wormhole can be articulated in Schwarzschild coordinates  $(t, r, \theta, \phi)$, is defined as follows: 
		\begin{equation}
			ds^2 = -e^{2\Phi(r)}dt^2 + \left(1 - \frac{b(r)}{r}\right)^{-1}dr^2 + r^2(d\theta^2 + \sin^2\theta d\phi^2),     \label{2.4}
		\end{equation}
		where $\Phi(r)$, $b(r)$ denotes the redshift function and shape function respectively. The values of radial the coordinate $r$ in between $(r_H \le r <\infty)$ . To avoid the formation of an event horizon $\Phi(r)$ must finite everywhere in spacetime. 
	\end{definition}
	\begin{definition}
		\cite{morris1988wormholes,sushkov2005wormholes} \label{definition 8.}
		The Morris-Thorne Flare-out Constraint constitutes for a wormhole throat to exist $b(r)$ obeys the flare-out condition $b(r_H) = r_H$  at boundary $r = r_H $ meaning the geometry must open outward rather than pinching into a singularity  and an outward-flaring derivative $b'(r_H) < 1$.
	\end{definition}
	\section{Results}
	\justifying
	Before proving our results, it is important to distinguish between our mathematical assumption and physical phenomena. Mathematically, we assume that our static, spherically symmetric spacetime manifold coupled to imperfect fluid admits an almost $\eta$-Ricci-Yamabe flow. Where the variables $ \alpha(r), \beta(r), \lambda(r) $ and $\omega(r)$ are defined in definition \ref{definition 3.} and this geometric flexibility ensures that the parameters depend strictly on the radial coordinate. This allows the geometric flow to dynamically adapt to the thermodynamic state of the surrounding imperfect fluid. Because of imperfect fluid characterized by anisotropic shear stress $\pi_{\mu\nu}$ and radial heat flux  $q_\mu $.\\
	Physically, the geometric flow acts dynamically in a way as the accretion fluid enters a dark energy phase at $\gamma = -1 $, its extreme negative radial pressure and anistropic friction supply the requisite exotic stress. Then functional parameters naturally adapt to these intense thermodynamic forces driving a structural reconfiguration of the spacetime. This interaction regularizes the apparent horizon forcing a spatial flare-out condition that de-singularizes the black hole geometry into a stable, traversable throat.\\
	
	Now, from the line element (\ref{1.2}), the non-vanishing temporal and radial components (coordinate differentials $dt^2$ and $dr^2$) of the metric tensor $g_{\mu\nu}$ are strictly defined by the structural function as:
	\[ g_{tt} = -f(r) \text{ and } g_{rr} =  \frac{1}{f(r)}.\]
	To construct the solitonic framework, we must compute 1-form $\eta_\mu$ imposing the strict geometric coupling defining it as the covariant dual of the radial vector field ($\eta_\mu = \xi_\mu $). The geometric flow by the given radial vector field $ \xi^{\mu} = (0, 1, 0, 0)$ then  1-form $ \eta_{t} $ corresponding temporal component trivially vanishes as:
	\begin{align*}
		\xi_t &= g_{tt}\xi^{t} \\
		\implies \eta_t &= 0.
	\end{align*}
	While, the 1-form $ \eta_{r} $ corresponding radial component evaluated as: 
	\begin{align}
		\xi_r &= g_{rr}\xi^{r}  \notag\\
		\implies \eta_r &= \frac{1}{f(r)}.  \label{3.1}
	\end{align}
	Next, we define the manifold deformation generated by this radial vector field $\xi = \partial_r$ i.e, with coordinate  components $ \xi^\lambda = \delta^\lambda_r$. Then the Lie derivative of the metric tensor along this flow \cite{debook, wald2010general} is formally expressed as:
	\begin{equation}
		\mathcal{L}_{\xi} g_{\mu\nu} = \xi^\lambda \partial_\lambda g_{\mu\nu} + g_{\lambda\nu} \partial_\mu \xi^\lambda + g_{\mu\lambda} \partial_\nu \xi^\lambda, \label{3.2}
	\end{equation}
	because we using specific vector field possesses strictly constant coordinate components ($\xi^\lambda = \delta^\lambda_r$), all partial derivatives of the vector field identically vanish ($\partial_\mu \xi^\lambda = \partial_\nu \xi^\lambda = 0$). Then  equation (\ref{3.2}) reduces the Lie derivative to a pure partial radial derivative of the metric components: 
	\begin{equation}
		\mathcal{L}_{\xi} g_{tt} = -f'(r), \quad \mathcal{L}_{\xi} g_{rr} = -\frac{f'(r)}{f(r)^2}. \label{3.3}
	\end{equation}
	
	\begin{theorem} \label{theorem1} 
		Let $(\mathcal{M}^4,g)$ be a 4-dimensional static, spherically symmetric spacetime manifold, characterized by the structural function $f(r)$ and coupled to an imperfect fluid, admitting an almost $\eta$-Ricci-Yamabe soliton along the radial vector field $\xi = \partial_r$. Then the scaling parameter $\omega(r)$ is strictly governed by the local gravitational lapse and the thermodynamic drag of the fluid according to:$$\omega(r) = f'(r) - \alpha(r) \Big[ f(r)^2 S_{rr} + S_{tt} \Big].$$
	\end{theorem}
	\begin{proof}
		To explicitly determine the functional parameters $\lambda(r)$ and $\omega(r)$, we systematically evaluate the components of the equation (\ref{2.1}).\\
		We begin by extracting the temporal component. For the radial vector field $\xi = \partial_r$ acting on the metric $g_{tt} = -f(r)$, the relevant geometric quantities evaluate strictly to $\mathcal{L}_{\xi}g_{tt} = -f'(r)$ and $\eta_t = 0$. Substituting these into equation (\ref{2.1}) yields:
		\begin{align}
			\frac{1}{2}\left(-f'(r)\right) + \alpha(r)  S_{tt} + \left(\lambda(r)  - \frac{\beta(r) \mathcal{R}}{2}\right)(-f(r)) + \omega(r)(0) &= 0 \label{3.4}\\
			\implies -\frac{f'(r)}{2} + \alpha(r) S_{tt} - \lambda(r)  f(r) + \frac{\beta(r) \mathcal{R} f(r)}{2} &= 0, \label{3.5}
		\end{align}
		dividing equation (\ref{3.5}) by $f(r)$ and rearranging to isolate the metric parameter, we obtain the explicit solution for  $\lambda$:
		\begin{equation}
			\alpha(r)  \frac{S_{tt}}{f(r)} - \frac{f'(r)}{2f(r)} + \frac{\beta(r) \mathcal{R}}{2} = \lambda(r). \label{3.6}
		\end{equation}
		Therefore, 
		\begin{equation}
			\lambda(r) = \alpha(r) \frac{S_{tt}}{f(r)} - \frac{f'(r)}{2f(r)} + \frac{\beta(r) \mathcal{R}}{2}. \label{3.7}
		\end{equation}
		Similarly, we determine the scaling parameter $\omega(r)$ via the radial component. Using the metric component $g_{rr} = f(r)^{-1}$, the geometric terms evaluate to $\mathcal{L}_{\xi}g_{rr} = - \frac{f'(r)}{f(r)^2}$ and $\eta_r = \frac{1}{f(r)}$. Substituting these into the equation (\ref{2.1}) provides:
		\begin{align}
			\frac{1}{2}\left(-\frac{f'(r)}{f(r)^2}\right) + \alpha(r) S_{rr} + \left(\lambda(r) - \frac{\beta(r) \mathcal{R}}{2}\right)\left(\frac{1}{f(r)}\right) + \omega(r)\left(\frac{1}{f(r)^2}\right) &= 0 \label{3.8}  \\
			\implies	-\frac{1}{2}f'(r) + \alpha(r) S_{rr} f(r)^2 + \left(\lambda(r) - \frac{\beta(r) \mathcal{R}}{2}\right)f(r) + \omega(r) &= 0  \label{3.9} \\
			\implies \hspace{1.24cm}	\frac{f'(r)}{2} - \alpha(r) S_{rr} f(r)^2 - \left(\lambda(r) - \frac{\beta(r) \mathcal{R}}{2}\right)f(r) &= \omega(r). \label{3.10}
		\end{align}
		Substitute, equation (\ref{3.7}) in equation (\ref{3.10}), then we obtain
		\begin{align}
			\omega(r) &= \frac{f'(r)}{2} - \alpha(r) S_{rr} f(r)^2 - \left( {\left[ \alpha(r) \frac{S_{tt}}{f(r)} - \frac{f'(r)}{2f(r)} + \frac{\beta(r) \mathcal{R}}{2} \right]} - \frac{\beta(r) \mathcal{R}}{2} \right)f(r) \notag \\
			\implies	\omega(r) &= \frac{f'(r)}{2} - \alpha(r) S_{rr} f(r)^2 - \left( \alpha(r) \frac{S_{tt}}{f(r)} - \frac{f'(r)}{2f(r)} \right)f(r).  \label{3.11}
		\end{align} 
		Therefore, \begin{equation}
			\omega(r) = f'(r) - \alpha(r) \Big[ S_{rr} f(r)^2 + S_{tt}  \label{3.12} \Big].
		\end{equation}
		Hence, concludes the proof of theorem \ref{theorem1}. The following corollary can be given as,
	\end{proof}
	
	\begin{corollary} \label{C1}
		Consider the manifold as established in theorem \ref{theorem1}. To ensure the radial functional parameter $\lambda(r)$ remains mathematically finite at the apparent horizon $f(r_H) = 0$, the temporal coordinate singularity must be regularized, strictly imposing the condition $\alpha(r_H) S_{tt}(r_H) = \frac{1}{2} f'(r_H)$. Consequently, this establishes a direct thermogeometric relation with the Hawking temperature ($T_H$), yielding $\alpha(r_H) S_{tt}(r_H) = 2\pi T_H$. This temporal regularization provides the foundational mathematical mechanism to de-singularize the manifold.
	\end{corollary}
	\begin{proof} 
		To evaluate the soliton at the apparent horizon ($r \to r_H$), we note that by the definition of an apparent horizon, the structural metric function vanishes, i.e., $\lim_{r \to r_H} f(r) = 0$ \cite{wald2010general}.\\
		From equation (\ref{3.5}), we obtain the radial functional parameter $\lambda(r)$:
		\begin{equation}
			\lambda(r) = \frac{2\alpha(r) S_{tt} - f'(r)}{2f(r)} + \frac{\beta(r) \mathcal{R}}{2}. \label{3.13}
		\end{equation}
		For the geometric flow to remain physically valid and continuous across the manifold, $\lambda(r_H)$ must be strictly finite. As $r \to r_H$, the denominator $f(r) \to 0$. Therefore, to prevent a geometric singularity, the numerator $2\alpha(r_H) S_{tt}(r_H) - f'(r_H)$ must simultaneously vanish at the exact coordinate $r_H$. By applying L'Hôpital's Rule to equation (\ref{3.13}) imposes the geometric constraint as:
		\begin{align}
			2\alpha(r_H) S_{tt}(r_H) - f'(r_H) &= 0   \label{3.14} \\
			\implies \alpha(r_H) S_{tt}(r_H) &= \frac{1}{2} f'(r_H). \label{3.15}
		\end{align}
		The surface gravity ($\mathcal{K}$), defined via Killing vector fields for a static, spherically symmetric spacetime, reduces at the apparent horizon to \cite{poisson2004relativist, carroll2004introduction}:
		\begin{equation}
			\mathcal{K} = \frac{1}{2} f'(r_H).     \label{3.16}
		\end{equation}
		Furthermore, Hawking established that the quantum thermal radiation emitted by the horizon (the Hawking Temperature, $T_H$) is strictly proportional to its surface gravity \cite{hawking1975particle}:
		\begin{align}
			T_H &= \frac{\mathcal{K}}{2\pi},   \label{3.17}  \\   
			\implies    \mathcal{K} &= 2\pi T_H.                 \label{3.18}
		\end{align}
		Substituting equations (\ref{3.16}) and (\ref{3.18}) into equation (\ref{3.15}), we obtain the exact thermogeometric equivalence:
		\begin{equation}
			\alpha(r_H) S_{tt}(r_H) = 2\pi T_H. \label{3.19}
		\end{equation}
		This completes the proof, demonstrating that the regularization of the apparent horizon explicitly couples the parameter $\alpha(r_H)$ to the Hawking temperature $T_H$.
	\end{proof}
	\begin{proposition} \label{P1}
		Let $(\mathcal{M}^4, g)$ be a 4-dimensional static, spherically symmetric spacetime manifold characterize by structural function $f(r)$ coupled to an imperfect fluid. Suppose that manifold admits an almost $\eta$-Ricci-Yamabe soliton along the radial vector field $\xi = \partial_r$. If the vector field is a gradient, such that $\xi = \nabla $, then the solitonic potential $\Psi$ satisfies the effective Poisson equation:
		\begin{equation}
			\Delta\Psi = (2\beta(r) - \alpha(r))\mathcal{R} - 4\lambda(r) - \frac{\omega(r)}{f(r)} \notag
		\end{equation}
		Consequently, the gradient soliton acts as an effective geometric potential field, where its Laplacian $\Delta\Psi$ directly corresponds to the kinematic expansion of the background manifold. 
	\end{proposition}
	\begin{proof} According to the Hessian Matrix, Lie derivative of the metric $g_{\mu\nu}$, when a vector field is a gradient ($\xi = \nabla \Psi$)  \cite{Petersen16} given as:
		\begin{equation}
			(\mathcal{L}_\xi g)_{\mu\nu} = 2 \nabla_{\mu} \nabla_{\nu} \Psi.  \label{3.20}
		\end{equation}
		Now, substitute the equation (\ref{3.20}) in the equation (\ref{2.1}), we get
		\begin{equation}
			\nabla_{\mu} \nabla_{\nu} \Psi + \alpha(r) S_{\mu\nu} + \left(\lambda(r) - \frac{\beta(r) \mathcal{R}}{2}\right)g_{\mu\nu} + \omega(r) \eta_{\mu} \eta_{\nu} = 0, \label{3.21}
		\end{equation}
		after taking the trace of the equation (\ref{3.21}), we get
		\begin{equation}
			\Delta \Psi + \alpha(r) \mathcal{R} + 4\lambda(r) - 2\beta(r) \mathcal{R} + \omega{\eta^v \eta_v} = 0. \label{3.22}
		\end{equation}
		Since, manifold characterize by lapse function $f(r)$, then from $ \eta^v \eta_v = \| \eta \|^2 $ where $\eta_\mu = \xi_\mu$ and $\xi_t = 0, \xi_r = 1/f(r)$, we have $\|\eta\|^2 = 1/f(r)$, and apply this to equation (\ref{3.22}) we obtained,
		\begin{equation}
			\Delta \Psi = (2\beta(r) - \alpha(r)) \mathcal{R} - 4\lambda(r) - \frac{\omega(r)}{f(r)}. \label{3.23}
		\end{equation}
		In classical physics gravitational potential function $ \Psi $ is governed by Poisson equation $\Delta \Psi = 4 \pi G \sigma $ where $\sigma$ is mass density.\\
		Furthermore, the kinematic expansion scalar $\Theta$ of a vector field is defined by its divergence: $\Theta \equiv \nabla_\mu \xi^\mu$. Substituting our gradient field definition ($\xi^\mu = \nabla^\mu \Psi$) yields $\Theta = \nabla_\mu \nabla^\mu \Psi \equiv \Delta\Psi$. \\
		Therefore, equation (\ref{3.23}) mathematically establishes that the gradient soliton acts as an effective geometric potential field, where its Laplacian $\Delta\Psi$ explicitly corresponds to the kinematic expansion of the background manifold.
	\end{proof}
	\begin{theorem}
		Consider the manifold as established in proposition \ref{P1}, which admits an almost $\eta$-Ricci-Yamabe soliton generated by a gradient vector field. If the manifold obeys barotropic Equation of State (EoS) $\rho = \gamma \sigma$, then the kinematic expansion parameter of the solitonic flow is strictly dictated by the thermodynamic trace of the fluid across three cosmological eras.
	\end{theorem}
	\begin{proof} In a 4-dimensional manifold, the traces of $g_{\mu\nu}, S_{\mu\nu}$, and $T_{\mu\nu} $ given as: $g^{\mu\nu}g_{\mu\nu} = 4$, $ g^{\mu\nu} S_{\mu\nu} = \mathcal{R}$, and  $ g^{\mu\nu} T_{\mu\nu} = T $. Since, the electromagnetic field is traceless then, $g^{\mu\nu}E_{\mu\nu} = 0$, by taking trace of Einstein field equation (\ref{2.3}), yields 
		\begin{equation}
			\mathcal{R} = 4\Lambda - \kappa T. \label{3.24}
		\end{equation} 
		Similarly, taking the trace of the energy-momentum tensor equation (\ref{2.2}) provides, $ T = 3\rho - \sigma $, and substituting into the equation (\ref{3.24}) gives:
		\begin{equation}
			\mathcal{R} = 4\Lambda - \kappa (3\rho - \sigma). \label{3.25}
		\end{equation}
		To express the curvature in terms of the mass density and thermodynamic state, we introduced the barotropic EoS $\rho = \gamma \sigma$ \cite{carroll2004introduction}. Substituting this relation into equation (\ref{3.25}), yields
		\begin{equation}
			\mathcal{R} = 4\Lambda - \kappa \sigma(3\gamma - 1). \label{3.26}
		\end{equation}
		As we established in Proposition \ref{P1}, the kinematic expansion of the gradient soliton is governed by the Poisson-like equation (\ref{3.23}), given as:
		\begin{equation}
			\Delta \Psi(r, \sigma) = (2\beta(r) - \alpha(r))\mathcal{R}(\sigma) - 4\lambda(r) - \frac{\omega(r)}{f(r)}, \label{3.27}
		\end{equation}
		also to isolate the fluid's thermodynamic driving force ($\sigma$), we set the radial coordinate at a specific depth ($r = r_0$) so that parameters $\alpha(r_0), \beta(r_0), \lambda(r_0),$ and $\omega(r_0)$ evaluate to fixed scalar constants, which makes the kinematic expansion $\Delta \Psi$ dependent on mass density and allowing us to accurately get the kinematic gradient $\frac{\partial (\Delta \Psi)}{\partial \sigma}$.\\
		By substituting equation (\ref{3.26}) into the isolated expansion relation (equation (\ref{3.27}) at fixed $r=r_0$), we can explicitly evaluate the solitonic response across three thermodynamic eras \cite{carroll2004introduction}.
		
		\begin{enumerate}
			\item Dust Cloud $(\gamma=0)$.\\
			For a pressureless fluid, the scalar curvature became additive $\mathcal{R} = 4\Lambda + \kappa\sigma$, yielding:
			\begin{equation}
				\Delta \Psi = (2\beta(r_0) - \alpha(r_0))\left[4 \Lambda + \kappa \sigma \right] - 4\lambda(r_0) - \frac{\omega(r_0)}{f(r_0)}.   \label{3.28}
			\end{equation}
			Here, the  kinematic gradient $\frac{\partial \Delta \Psi}{\partial \sigma} = \kappa(2\beta(r_0) - \alpha(r_0))$, then the solitonic flow scales directly against the gravitational contraction of the dust, filtered through the local Ricci-Yamabe coupling parameters.
			\item Radiation Era ($\gamma = \frac{1}{3}$).\\
			For a relativistic radiation fluid, the mass density is completely vanish from the scalar curvature ($\mathcal{R} = 4\Lambda$), yielding: 
			\begin{equation}
				\Delta \Psi = (2\beta(r_0) - \alpha(r_0))\big[4\Lambda\big] - 4\lambda(r_0) - \frac{\omega(r_0)}{f(r_0)}.  \label{3.29}
			\end{equation}
			Here kinematic gradient \(\frac{\partial \Delta \Psi}{\partial \sigma} = 0\), which means the geometric expansion is unaffected by fluid density. Consequently, the solitonic flow decouples from the matter field and propagates through the irradiated manifold as if it were a pure, empty de Sitter vacuum.
			\item Dark Energy / Exotic Matter ($\gamma = -1$). \\
			For an exotic fluid, the extreme negative pressure violates the NEC. The curvature $\mathcal{R} = 4\Lambda + 4\kappa\sigma$ violently accelerates, yielding:
			\begin{equation}
				\Delta \Psi = (2\beta(r_0) - \alpha(r_0))\big[4\Lambda + 4 \kappa\sigma\big] - 4\lambda(r_0) - \frac{\omega(r_0)}{f(r_0)}.  \label{3.30}
			\end{equation}
			Here, the kinematic gradient $\frac{\partial\Delta\Psi}{\partial\sigma} = 4\kappa(2\beta(r_0) - \alpha(r_0))$. The extreme negative effective pressure resulting from the NEC vilation provides the strict theoretical stress requisite to satisfy the flare-out condition, formally supporting the geometric transition to a traversable wormhole topology.\\
		\end{enumerate}
		\begin{figure}[h]
			\includegraphics[width=1.0\linewidth]{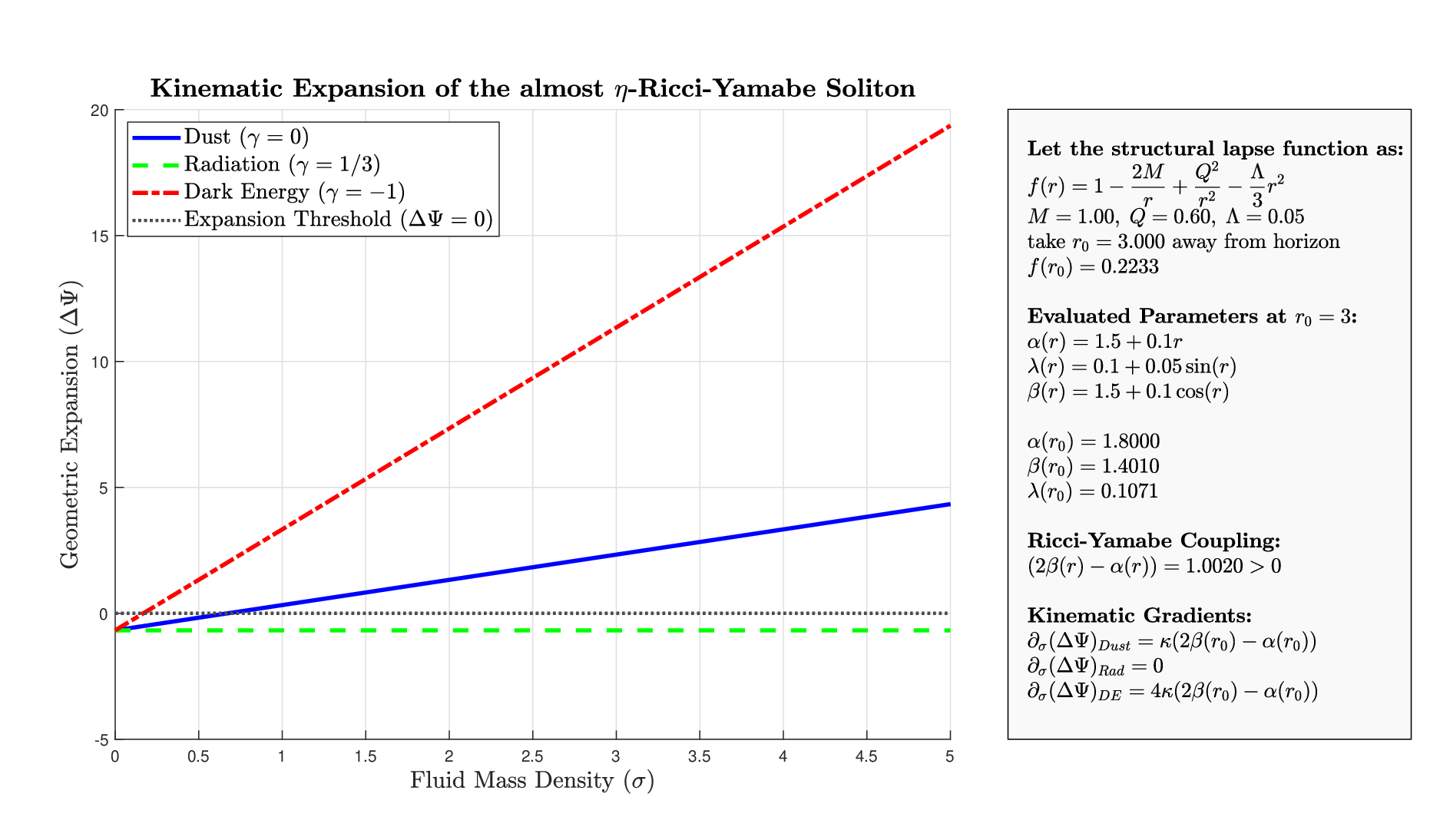}
			\caption{Shows kinematic expansion ($\Delta \Psi$) of the almost $\eta$-Ricci-Yamabe soliton against accretion fluid density ($\sigma$). The different types of matter resist gravity in distinct ways: radiation $\gamma= \frac{1}{3} $ remains unaffected, and dust $\gamma = 0 $ provides linear resistance. While dark energy regime ($\gamma=-1$) utilizes negative effective pressure to generate an amplified gradient $4\kappa(2\beta(r_0)-\alpha(r_0))$. This geometric expansion strictly dominates the attractive curvature gradient ($\Delta\Psi > 0$), mathematically precluding complete gravitational collapse and establishing the necessary conditions for a topological transition.}
			\label{fig:era}
		\end{figure}
		In the Fig. \ref{fig:era}, we choose RNdS lapse function $f(r)$, with mass $M=1.00$, Charge $Q=60$, cosmological constant $\Lambda = 0.05$ and evaluate the flow at $r_0 = 3$, a safe distance from the black hole geometry to observe how the normal fluid behaves in the bulk spacetime. Here, the geometric friction parameter $\alpha(r)$ increases linearly with distance, while $\beta(r)$ and $\lambda(r)$ are modeled as smooth oscillating functions to represent the natural ripples of the accretion disk. we choose these specific values to verify the Ricci-Yamabe constraint $2\beta(r_0) - \alpha(r_0)$. It is strictly positive, and mathematically required for the geometric expansion to scale correctly with the fluid's density.
	\end{proof}
	\begin{theorem} \label{Theorem3} 
		Let $(\mathcal{M}^4, g)$ be a 4-dimensional static, spherically symmetric spacetime manifold coupled to an imperfect fluid, admitting an almost $\eta$-Ricci-Yamabe soliton along radial vector field $\xi= \partial_r$. If the fluid's stress-energy tensor violates the NEC at the horizon locus $r = r_H$, then the geometric flow generates a strictly positive solitonic scaling factor $\omega(r_H) > 0$, which mathematically forces the spatial geometry to satisfy the Morris-Thorne flare-out condition $b'(r_H) < 1$.
	\end{theorem}
	\begin{proof} Consider the 4-dimensional manifold defined in definition \ref{def 3}. By virtue of Definition \ref{definition 6.}, the violation of the NEC fundamentally implies the existence of exotic matter, characterized by a negative effective pressure.\\
		Then, we construct a purely radial null vector field $k^\mu$ on the manifold, which takes the form $k^\mu = (k^t, k^r, 0, 0)$. Since $k^\mu$ is a null vector, its invariant norm must vanish identically, yielding the condition: $g_{\mu\nu} k^\mu k^\nu = 0$. Then,
		\begin{align*}
			g_{tt} k^t k^t + g_{rr} k^r k^r &= 0\\
			\implies -f(r)(k^t)^2 + \frac{1}{f(r)} (k^r)^2 &= 0 \\
			\implies \hspace{2.68cm} (k^r)^2 &= f(r)^2 (k^t)^2.
		\end{align*}
		By choosing \( k^t = \frac{1}{\sqrt{f(r)}} \) normalize the temporal component,
		and take positive root, \(  k^r = \sqrt{f(r)} \). From the Einstein field equation violation of NEC mathematically translates to evaluating the curvature term $S_{\mu\nu} k^\mu k^\nu$, yields
		\begin{align}
			S_{\mu\nu} k^\mu k^\nu &= S_{tt}(k^t)^2  + S_{rr} (k^r)^2 \notag\\
			\implies \hspace{0.8cm} S_{\mu\nu} k^\mu k^\nu  &= S_{tt}\frac{1}{f(r)}  + S_{rr} f(r) \notag \\
			\implies S_{rr} f(r)^2 + S_{tt} &= f(r) \Big( S_{\mu\nu} k^\mu k^\nu \Big). \label{3.31}
		\end{align}
		To ensure this metric association represents a true topological transformation, the mapping to the Morris-Thorne topology is physically valid if and only if two strict conditions are satisfied simultaneously. First, the solitonic temporal regularization must ensure the redshift function $\Phi(r)$ remains everywhere finite (satisfying definition \ref{definition 7.}) to prevent the formation of an absolute horizon.\\
	    Second, to satisfy the mandatory spatial flare-out condition (definition \ref{definition 8.}). For this condition, we must show that the localized stress-energy of the imperfect fluid into the dark energy regime at $\gamma = -1$, dynamically forces the generic structural metric $f(r)$ to manifest as a radially continuous shape function $b(r)$.\\
		Therefore, to equate our structural metric to the Morris-Thorne framework, we must first map their respective spatial geometries. In standard Morris-Thorne metric, the radial component is defined as:
		\begin{equation}
			g_{rr} = \frac{1}{1 - \frac{b(r)}{r}}, \label{3.32}
		\end{equation}
		on the other hand, from the spacetime metric (\ref{1.2}), the radial component is given by:
		\begin{equation}
			g_{rr} = \frac{1}{f(r)}.  \label{3.33}
		\end{equation}
		Equating equations (\ref{3.32}) and (\ref{3.33}), we obtain the relation:
		\begin{equation}
			f(r) = 1 - \frac{b(r)}{r}. \label{3.34}
		\end{equation}
		At the wormhole throat $r = r_H$, the shape function must satisfy the essential condition $b(r_H) = r_H$. Consequently, evaluating equation (\ref{3.34}) at the throat boundary, yields:
		\begin{equation}
			f(r_H) = 1 - \frac{r_H}{r_H} = 0. \label{3.35}
		\end{equation}
		 By taking derivative of equation (\ref{3.34}) with respect to the radial coordinate $r$, we find:
		\begin{equation}
			f'(r) = \frac{b(r) - r b'(r)}{r^2}.  \label{3.36}
		\end{equation}
		Evaluating equation (\ref{3.35}) at the wormhole throat $r = r_H$, and substituting the geometric constraint $b(r_H) = r_H$, we obtain:
		\begin{equation}
			f'(r_H) =  \frac{1 - b'(r_H)}{r_H}. \label{3.37}
		\end{equation}
		Furthermore, by synthesizing equations (\ref{3.12}) and (\ref{3.31}), the solitonic scaling factor is defined as:
		\begin{equation}
			\omega(r) = f'(r) - \alpha(r) \Big[ f(r) (S_{\mu\nu} k^\mu k^\nu) \Big], \label{3.38}
		\end{equation}
		Also for a standard, non-extremal static, spherically symmetric black hole, the lapse function $f(r)$ is negative inside the horizon and positive outside. Then the slope of the function at $r = r_H$ must point strictly upward, ensuring an inherently positive metric gradient $f'(r_H) > 0$. Therefore, equation (\ref{3.38}) mathematically guarantees that the geometric flow expands at the boundary:
		\begin{align}
			\omega(r_H) &= f'(r_H) > 0 \notag \\
			\implies \omega(r_H) &> 0. \label{3.39}
		\end{align}
		Finally, substituting $f'(r_H) > 0$, back into equation (\ref{3.37}), and recognizing the physical requirement of a strictly positive throat radius $r_H > 0$, mathematically forces the numerator to be strictly greater than zero:
		\begin{align}
			1 - b'(r_H) &> 0 \notag \\
			\implies b'(r_H) &< 1. \notag
		\end{align}
		This concludes the proof, explicitly demonstrating that the endogenous geometric flow dynamically forces the spatial geometry to satisfy the mandatory Morris-Thorne flare-out condition without relying on any a priori shape assumptions.\\
	\end{proof}
		\begin{figure}[h]
		\centering
		\includegraphics[width=1\linewidth]{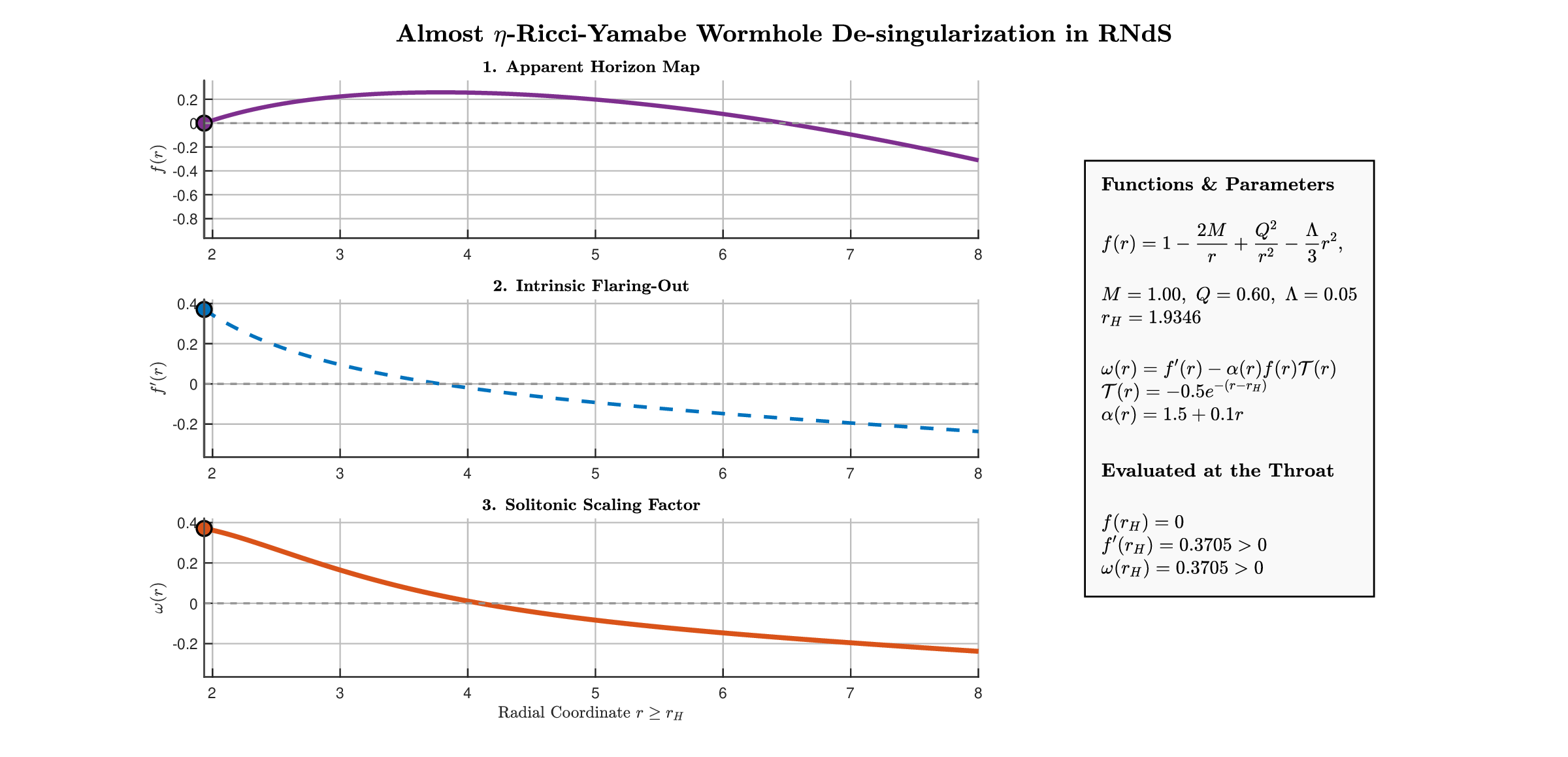}
		\caption{Shows that the top graph explain the spacetime metric smoothly opening from the wormhole throat while the middle graph confirms the flaring-out condition, mathematically proving the geometry curves outward to remain open rather than collapsing into a singularity and the bottom graph reveals that the almost $\eta$-Ricci-Yamabe soliton providing the strictly positive scaling factor $\omega(r_H) > 0$ required at the throat to prevent gravitational collapse. As the radial distance increases, the curvature and expansion fields naturally transition into the negative domain. This verifies that the solitons exotic geometric effects are strictly localized to the throat, asymptotically recovering the standard de Sitter $\Lambda$ expansion at spatial infinity.}
		\label{fig:wthm2}
	\end{figure}
	\noindent In Fig. \ref{fig:wthm2}, we choose RNdS lapse function $f(r)$, with ($M, Q, \Lambda$), as taken in figure \ref{fig:era}. Since at apparent horizon where the wormhole throat must form, $ f(r_H) =0$ yields apparent horizon $r_H \approx 1.93$. Also choose the exotic matter trace $\mathcal{T}(r) = -0.5e^{-(r-r_H)}$, as an exponential decay so that its negative pressure is strong enough to push the throat open, but drops off quickly to prevent it from leaking into normal space. Since $f(r_H)=0$, so that this specific matter trace produces a strictly positive scaling factor $\omega(r_H) > 0$, mathematically proving the throat successfully flares open.
	\begin{corollary} \label{C2}
		Let $(\mathcal{M}^4, g)$ be a 4-dimensional static, spherically symmetric manifold coupled to an imperfect fluid, admitting an almost $\eta$-Ricci-Yamabe soliton along radial vector field $\xi= \partial_r$. Then the geometric flaring-out condition  established in Theorem \ref{Theorem3} requires the spatial expansion parameter at the emergent throat $r_H$ to satisfy the strict algebraic constraint,
		$$ \omega(r_H) = \frac{2M}{r_H^2} - \frac{2Q^2}{r_H^3} - \frac{2\Lambda}{3}r_H > 0, $$
		where $M, Q, \Lambda $ is denotes the mass, electric charge, and cosmological constant of the spacetime respectively.\\
	\end{corollary}
	\begin{proof} For the RNdS manifold, the structural lapse function defined as (\ref{1.3}) and its derivative is
		\begin{equation}
			f'(r) = \frac{2M}{r^2} - \frac{2Q^2}{r^3} - \frac{2\Lambda}{3}r, \label{3.40}
		\end{equation}
		at the throat $r = r_H$ equation (\ref{3.40}) yields:
		\begin{align}
			f'(r_H) &= \frac{2M}{r_H^2} - \frac{2Q^2}{r_H^3} - \frac{2\Lambda}{3}r_H \label{3.41}  \\
			\implies \omega(r_H) &= \frac{2M}{r_H^2} - \frac{2Q^2}{r_H^3} - \frac{2\Lambda}{3}r_H > 0. \label{3.42}
		\end{align}
		The equation (\ref{3.42}) proves that the solitonic scaling factor $\omega(r_H)$ must inherently dominate the attractive gravitational mass gradient $\left( \frac{2M}{r_H^2}  \right)$, by taking mathematical advantage of outward push from the charge and cosmic expansion terms $\left(-\frac{2Q^2}{r_H^3} - \frac{2\Lambda}{3}r_H \right)$.
	\end{proof}

	\begin{corollary} \label{C3}
		Consider the 4-dimensional manifold $(\mathcal{M}^4, g)$, as established in Theorem \ref{Theorem3}. At the throat locus $r = r_H$, the radial derivative of the solitonic scaling factor strictly dominates the local geometric curvature gradient, satisfying the strict inequality $\omega'(r_H) > f''(r_H)$.
	\end{corollary}
	\begin{proof}
		By contracting the modified Einstein field equations (\ref{2.3}) with  null vector $k^\mu$ the metric trace components perfectly vanish due to the null condition $g_{\mu\nu}k^\mu k^\nu = 0$. This establishes a strict geometric equivalence between the Ricci curvature and the total effective stress-energy tensor given as:
		\begin{equation}
			S_{\mu\nu}k^\mu k^\nu = \kappa (T_{\mu\nu} + E_{\mu\nu}) k^\mu k^\nu.  \label{3.43}
		\end{equation}
		Now, evaluate the violation of the NEC entirely through geometry of the manifold. Define equation (\ref{3.43}) total geometric matter-trace equivalent as $\mathcal{T}(r) \equiv S_{\mu\nu}k^\mu k^\nu$.
		From the theorem \ref{Theorem3} equation (\ref{3.38}) for the almost $\eta$-Ricci-Yamabe flow is given as:
		\begin{equation}
			\omega(r) = f'(r) - \alpha(r) f(r) \mathcal{T}(r),   \label{3.44}
		\end{equation}
		also to evaluate the radial derivative of the solitonic scaling factor, differentiating equation (\ref{3.44}) with respect to the radial coordinate $r$, yields
		\begin{equation}
			\omega'(r) = f''(r) - \Big[ \alpha'(r) f(r) \mathcal{T}(r) + \alpha(r) f'(r) \mathcal{T}(r) + \alpha(r) f(r) \mathcal{T}'(r) \Big],   \label{3.45}
		\end{equation}
		at the de-singularized apparent horizon $r = r_H$ and by theorem \ref{Theorem3} equation (\ref{3.35}) at the boundary $f(r_H) = 0$, then equation (\ref{3.45}) annihilates the terms containing $f(r_H)$ we obtain
		\begin{equation}
			\omega'(r_H) = f''(r_H) - \alpha(r_H) f'(r_H) \mathcal{T}(r_H) \label{3.46}.
		\end{equation}
		Let us defined, $\mathcal{N}(r_H) \equiv - \alpha(r_H) f'(r_H) \mathcal{T}(r_H)$. According to Theorem \ref{Theorem3}, we hold that $\alpha(r_H) > 0$ and $f'(r_H) > 0$. Furthermore, the violation of the NEC explicitly dictates that $\mathcal{T}(r_H) < 0$. Therefore, $\mathcal{N}(r_H) > 0$, then it mathematically follows that $\omega'(r_H) = f''(r_H) + \mathcal{N}(r_H)$. Consequently, we obtain the strict inequality:
		\begin{equation}
			\omega'(r_H) > f''(r_H). \label{3.47}
		\end{equation}
		Here, the radial derivative of the solitonic scaling factor strictly dominates the local geometric curvature gradient. It mathematically confirm structural repulsion immediately outside the throat. It required to ensure the traversability and stability.
	\end{proof}
	
	\begin{corollary}\label{C4}
		Consider the traversable wormhole solution within the 4-dimensional manifold $(\mathcal{M}^4, g)$ as established in Theorem~\ref{Theorem3}. Assuming the localized exotic matter trace decays exponentially at spatial infinity, the solitonic scaling factor $\omega(r)$ naturally recovers the intrinsic geometric gradient $f'(r)$. Consequently, we obtain:
		$\displaystyle\lim_{r \to \infty} \omega(r) = \displaystyle\lim_{r \to \infty} f'(r),$ which universally ensures that the soliton is cosmologically safe and compatible with both asymptotically flat and RNdS spacetime.
	\end{corollary}

	\begin{proof}
		To evaluate the asymptotic behavior of the almost $\eta$-Ricci-Yamabe  flow in the limit as the radial coordinate approaches spatial infinity ($r \to \infty$). From equation (\ref{3.43}), limit of the solitonic scaling factor is given by:
		\begin{equation}
			\lim_{r \to \infty} \omega(r) = \lim_{r \to \infty} \Big[ f'(r) - \alpha(r) f(r) \mathcal{T}(r)\Big]. \label{3.48}
		\end{equation}
		Now, analyze the limiting behavior of  $(-\alpha(r) f(r) \mathcal{T}(r))$. To achieve wormhole stability, the exotic matter requires to de-singularize the manifold and strictly localized within the high-curvature domain of the throat. As $r \to \infty$ the thermodynamic fluid density decays to  pure vacuum state enforcing the exponential decay of the geometric matter-trace equivalent:
		\begin{equation}
			\lim_{r \to \infty} \mathcal{T}(r) = 0.  \label{3.49}
		\end{equation}
		Also provided that the $\alpha(r)$ and $f(r)$ exhibit at most polynomial growth, the exponential decay of $\mathcal{T}(r)$ strictly dominates. It trivially follows that
		\begin{equation}
			\lim_{r \to \infty} \Big( -\alpha(r) f(r) \mathcal{T}(r) \Big) = 0, \label{3.50}
		\end{equation}
		substituting equation (\ref{3.50}) into (\ref{3.43}), the almost $\eta$-Ricci-Yamabe flow mathematically reduces to the intrinsic geometric gradient:
		\begin{equation}
			\lim_{r \to \infty } \omega(r) = \lim_{r \to \infty} f'(r) - 0 =  \lim_{r \to \infty } f'(r). \label{3.51}
		\end{equation}
		From the universal geometric equivalence in (\ref{3.51}), two standard cosmological limits naturally emerge:
		\begin{enumerate}
			\item \textbf{Asymptotically Flat Backgrounds:} For standard static, spherically symmetric metrics  the metric flattens to Minkowski space ($f(r) \to 1$). Consequently, the intrinsic geometric slope vanishes $\displaystyle\lim_{r \to \infty} f'(r) = 0$, which enforces $\displaystyle\lim_{r \to \infty} \omega(r) = 0$.
			\item \textbf{Asymptotically de Sitter Backgrounds (RNdS):} For a universe governed by a cosmological constant $\Lambda$, the background metric is (\ref{1.3}) at spatial infinity the mass and charge terms vanish leaving the intrinsic cosmological expansion \( \displaystyle\lim_{r \to \infty} f'(r) = -\frac{2\Lambda}{3}r \) . Then the soliton yields entirely to dark energy enforcing $\displaystyle\lim_{r \to \infty} \omega(r) = -\frac{2\Lambda}{3}r$.
		\end{enumerate}
		This mathematically establishes that the almost $\eta$-Ricci-Yamabe flow serves as a  localized topological modifier, While it effectively prevents gravitational collapse at the throat at $r = r_H$, also leaves the deep-space global topology of the universe completely undisturbed.
	\end{proof}
	\begin{figure}[h]
		\centering
		\includegraphics[width=1\linewidth]{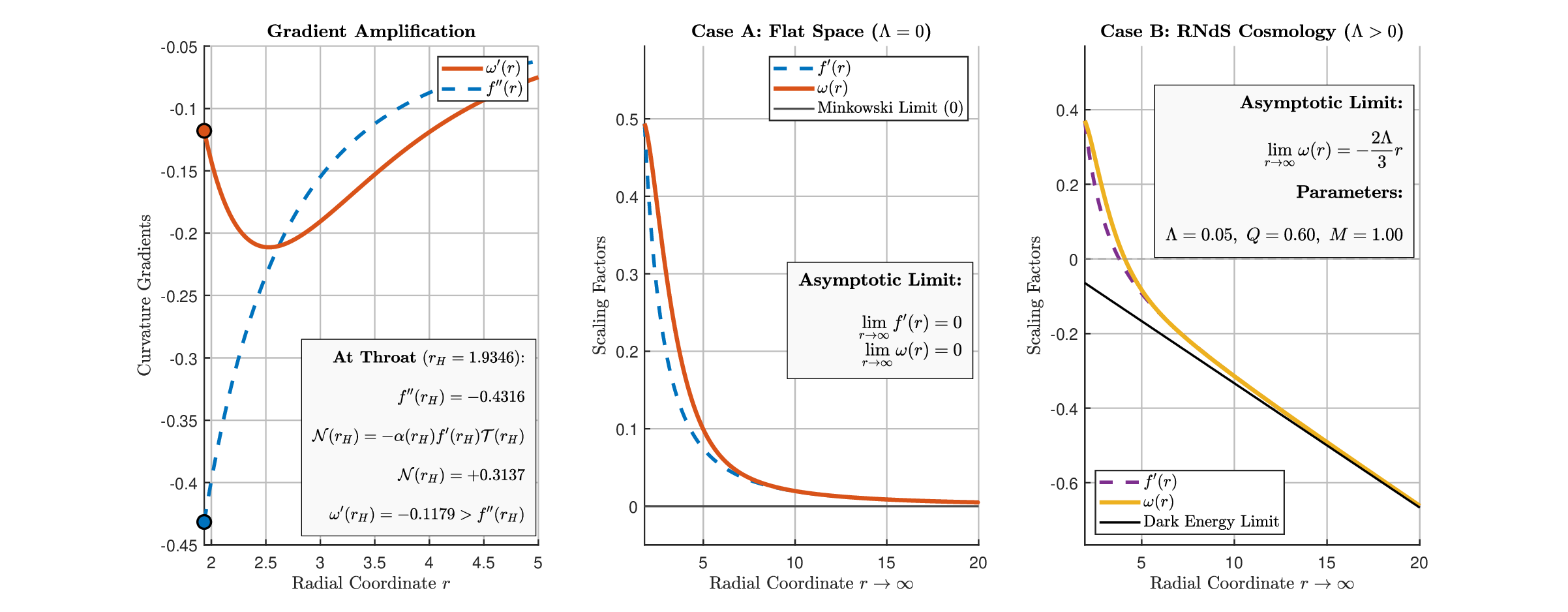}
		\caption{This figure proves the wormhole is both locally stable and asymptotically safe. Left-hand side graph explains the soliton successfully amplifies the curvature gradient ($\omega^{\prime}(r_H) > f^{\prime\prime}(r_H) $) exactly at  throat providing the necessary outward force to prevent gravitational collapse. The middle and right-hand side graphs confirm this topological modification is strictly localized, smoothly decaying to zero in a flat space (Minkowski) scenario, or asymptotically converging to the background expansion of the RNdS cosmology.}
		\label{fig:wcorrollary}
	\end{figure}
	\noindent In Fig. \ref{fig:wcorrollary}, also the RNdS lapse function $f(r)$, with $M,Q$, and $\alpha(r)$ as taken in figure \ref{fig:era}. We evaluated exactly at the wormhole throat at $r_H = 1.9346$ to mathematically map the locus of the topological transition. From corollary \ref{C3} the parameters at throat are $f''(r_H) = -0.4316$, and  $\mathcal{N}(r_H) = 0.3137$, yielding scaling gradient $\omega'(r_H) = -0.1179$, this rigorously satisfies the strict geometric inequality $\omega'(r_H) > f''(r_H)$. While the scaling factor $\omega(r)$ evaluated far away from throat when $\Lambda=0$ and $0.05$ to verify the asymptotic safety of the metric in both Minkowski and expanding de Sitter spacetimes.
	
	\begin{theorem} \label{theorem4}
		Consider the manifold as established in Theorem \ref{Theorem3}. Under the geometric flow of the almost $\eta$-Ricci-Yamabe soliton, the violation of the NEC by the fluid forces the spacetime manifold to undergo a topological transition into a traversable wormhole. This transition is mathematically confirm by the simultaneous de-singularization of the temporal coordinate and the flaring-out of the spatial geometry.
	\end{theorem}
	\begin{proof}
		To rigorously show a complete topological transition from a black hole geometry to a traversable wormhole, the manifold must simultaneously satisfy two independent geometric criteria at the apparent horizon $r_H$, given as:
		\begin{enumerate}
			\item Temporal De-singularization Condition: \\
			To prove the prevention of an event horizon (Ensuring the redshift function $\Phi(r)$ remains finite), the time component of the metric must be non-zero $g_{tt} \ne 0$ across the boundary, meaning time does not freeze. We evaluate the $tt$-component of the almost $\eta$-Ricci-Yamabe flow equation (\ref{2.1}):
			\begin{equation}
				(\mathcal{L}_\xi g)_{tt} + \alpha S_{tt} + \left(\lambda(r) - \frac{\beta\mathcal{R}}{2}\right)g_{tt} + \omega(r) \eta_t \eta_t = 0. \notag
			\end{equation}
			Since, the flow is along radial gradient vector field $\xi = \partial_r$, its temporal component vanishes identically $\xi^t = 0$. Consequently, the 1-form $\eta_t = g_{tt}\xi^t = 0$, yielding:
			\begin{equation}
				(\mathcal{L}_\xi g)_{tt} + \alpha S_{tt} + \left(\lambda(r) - \frac{\beta\mathcal{R}}{2}\right)g_{tt} = 0. \label{3.52}
			\end{equation}
			Now, by using the standard Morris-Thorne redshift metric, $g_{tt} = -e^{2\Phi(r)}$. Taking the Lie derivative of this metric component along the radial vector yields:
			\begin{equation}
				(\mathcal{L}_\xi g)_{tt} = \partial_r \left( -e^{2\Phi(r)} \right) = -2\Phi'(r)e^{2\Phi(r)}. \label{3.53}
			\end{equation}
			Substituting equation (\ref{3.53}) into the flow equation (\ref{3.52}) and dividing the entire relation by the temporal metric component ($g_{tt} = -e^{2\Phi(r)}$) isolates the spatial gradient of the redshift function, we obtain:
			\begin{equation}
				\Phi'(r) = \frac{1}{2} \left[ \alpha \frac{S_{tt}}{g_{tt}} - \lambda(r) + \frac{\beta\mathcal{R}}{2} \right]. \label{3.54}
			\end{equation}
			As we established in Corollary \ref{C1} via L'Hôpital's rule that the $\lambda(r_H)$ strictly finite limit exactly at the apparent horizon boundary. Furthermore, the temporal component of the curvature tensor $S_{tt}$ and the scalar curvature $\mathcal{R}$ remain strictly continuous. \\
			Crucially, the ratio $\frac{S_{tt}}{g_{tt}}$ maps directly to the mixed curvature tensor $S^t_t$, this term evaluates to a strictly finite scalar and cannot diverge. This gives the entire right side of equation (\ref{3.54}) evaluates to a finite scalar, yields:
			\begin{equation}
				\lim_{r \to r_H} \Phi'(r) = \text{Finite}. \label{3.55}
			\end{equation}
			By the fundamental theorem of calculus, a strictly finite spatial derivative across a boundary guarantees that the function itself is continuous and finite.\\
			Therefore, the temporal metric component evaluates to a non-zero constant ($g_{tt} = -e^{2\Phi} \neq 0$).\\
			\item The prevention of spatial collapse (The Morris-Thorne flare-out condition).\\
			Simultaneously, by virtue of Theorem \ref{Theorem3}, the extreme dark energy state ($\gamma = -1$) of the localized fluid dynamically generates a strictly positive solitonic scaling factor exactly at the throat $\omega(r_H) > 0$. This positive geometric expansion mathematically forces the spatial geometry to strictly satisfy the mandatory Morris-Thorne flare-out condition $b'(r_H) < 1$, ensuring the geometry remains open.\\
		\end{enumerate}
		Consequently, both the flare-out and temporal requirements are simultaneously mathematically satisfies the necessary and sufficient framework to completely transform the initial black hole geometry into a stable, traversable wormhole throat.
	\end{proof}
	\begin{theorem} \label{theorem5}
		Let $(\mathcal{M}^4, g)$ be a 4-dimensional static, spherically symmetric spacetime manifold coupled to an imperfect fluid, admitting an almost $\eta$-Ricci-Yamabe soliton along radial vector field $\xi= \partial_r$. Under first-order  linear metric perturbations, the geometric flow introduces a localized dissipative mechanism. Provided $\alpha(r) > 0$, the perturbation evolution reduces to a damped wave equation, thereby imposing a strict geometric drag on the manifold.
	\end{theorem}
	\begin{proof} Let us consider the static, spherically symmetric spacetime manifold, which is perfectly smooth and characterized by covariant metric tensor $g_{\mu\nu}$, and inverse metric tensor $g^{\mu\nu}$. We introduce a first-order gravitational wave perturbation ($h_{\mu\nu}$), is tracked by a microscopic parameter $\epsilon \ll 1$, such that the perturbed metric tensor takes the form:
		\begin{equation}
			\tilde{g}_{\mu\nu} = g_{\mu\nu} + \epsilon h_{\mu\nu}, \label{3.56}
		\end{equation}
		to strictly satisfy the condition $\tilde{g}_{\mu l} \tilde{g}^{{l}\nu} = \delta_\mu^\nu$, we assume the perturbed inverse metric tensor to be of the form:
		\begin{equation}
			\tilde{g}^{\mu\nu} = g^{\mu\nu} + \epsilon A^{\mu\nu}. \label{3.57}
		\end{equation} 
		Substituting equations (\ref{3.56}) and (\ref{3.57}) into the condition $\tilde{g}_{\mu{l}} \tilde{g}^{{l}\nu} = \delta_\mu^\nu$, linearizing the expression, and neglecting higher-order terms of $\epsilon^2$, we obtain
		\begin{align}
			(g_{\mu{l}} + \epsilon h_{\mu{l}})(g^{{l}\nu} + \epsilon  A^{{l}\nu}) &= \delta_\mu^\nu \notag \\
			\implies \hspace{1.50cm} \delta_\mu^\nu + \epsilon(A_\mu^\nu + h_\mu^\nu) &= \delta_\mu^\nu \notag \\
			\implies \hspace{3.40cm} A_\mu^\nu &= -h_\mu^\nu.  \notag
		\end{align}
		Substitute $ A_\mu^\nu = -h_\mu^\nu$  into equation (\ref{3.57}), then the inverse perturbed metric tensor defined as:
		\begin{equation}
			\tilde{g}^{\mu\nu} = g^{\mu\nu} - \epsilon h^{\mu\nu}. \label{3.58}
		\end{equation}
		Similarly, the affine connections (Christoffel symbols) undergo a first-order perturbation, defined as: $\tilde{\Gamma}^\varrho_{\mu\nu} = \Gamma^\varrho_{\mu\nu} + \epsilon \delta\Gamma^\varrho_{\mu\nu}$. According to fundamental theorem of Riemannian geometry, the metric is covariantly constant with respect to its own connection, yield
		\begin{align}
			\tilde{\nabla}_{l} \tilde{g}_{\mu\nu} &= 0 \notag \\
			\implies \hspace{2.83cm} \partial_{l} \tilde{g}_{\mu\nu} - \tilde{\Gamma}^{s}_{{l}\mu} \tilde{g}_{{s}\nu} - \tilde{\Gamma}^{s}_{{l}\nu} \tilde{g}_{\mu{s}} &= 0 \notag \\
			\implies \partial_{l} \tilde{g}_{\mu\nu} - (\Gamma^{s}_{{l}\mu} + \epsilon \delta\Gamma^{s}_{{l}\mu}) \tilde{g}_{{s}\nu} - (\Gamma^{s}_{{l}\nu} + \epsilon \delta\Gamma^{s}_{{l}\nu}) \tilde{g}_{\mu{s}} &= 0,           \label{3.59}
		\end{align}
		applying the condition $\tilde{g}_{\mu\nu} = g_{\mu\nu} + \epsilon h_{\mu\nu}$ to equation (\ref{3.59}) and strictly neglecting $\epsilon^2$ terms, we get
		\begin{equation}
			\nabla_{l} (g_{\mu\nu} + \epsilon h_{\mu\nu}) - \epsilon \delta\Gamma^{s}_{{l}\mu} g_{{s}\nu} - \epsilon \delta\Gamma^{s}_{{l}\nu} g_{\mu{s}} = 0, \label{3.60}
		\end{equation}
		imposing metric tensor compatibility ($\nabla_{l} g_{\mu\nu} = 0$) and lowering the connection index ($g_{{s}\nu} \delta\Gamma^{s}_{{l}\mu} = \delta\Gamma_{\nu{l}\mu}$) simplifies equation (\ref{3.60}) to:
		\begin{equation}
			\nabla_{l} h_{\mu\nu} = \delta\Gamma_{\nu{l}\mu} + \delta\Gamma_{\mu{l}\nu}, \label{3.61}
		\end{equation}
		permuting the indices ($\mu, \nu, {l}$), circularly three times in equation (\ref{3.61}), and using $\delta\Gamma_{{l}\mu\nu} = \delta\Gamma_{{l}\nu\mu}$ provides:
		\begin{equation}
			\nabla_\mu h_{\nu{l}} + \nabla_\nu h_{{l}\mu} - \nabla_{l} h_{\mu\nu} = 2 \delta\Gamma_{{l}\mu\nu}, \label{3.62}
		\end{equation}
		contracting equation (\ref{3.62}) with the inverse metric $g^{\varrho{l}}$ isolates the first-order variation of the Christoffel symbols:
		\begin{equation}
			\delta\Gamma^\varrho_{\mu\nu} = \frac{1}{2} (\nabla_\mu h_{\nu}^{\varrho} + \nabla_\nu h_{\mu}^{\varrho} - \nabla^{\varrho} h_{\mu\nu}). \label{3.63}
		\end{equation}
		Similarly, contracting the equation (\ref{3.62}), indices $\nu$ and $\varrho$ we get the trace-like variation:
		\begin{equation}
			\delta\Gamma^\varrho_{\mu\varrho} = \frac{1}{2} (\nabla_\mu h_{\varrho}^{\varrho} + \nabla_\varrho h_{\mu}^{\varrho} - \nabla^{\varrho} h_{\mu\varrho}), \label{3.64}  
		\end{equation}
		from the index symmetry $\nabla_\varrho h_{\mu}^{\varrho} = \nabla^{\varrho} h_{\mu\varrho}$ and defining the wave perturbation trace as $h \equiv h^{\varrho}_{\varrho}$, in equation (\ref{3.64}), yield 
		\begin{equation}
			\delta\Gamma^\varrho_{\mu\varrho} = \frac{1}{2} (\nabla_\mu h). \label{3.65}
		\end{equation}
		The Palatini Identity \cite{palatini1919deduzione} dictates the variation in Ricci curvature tensor as follows:
		\begin{equation}
			\delta S_{\mu\nu} = \nabla_\varrho (\delta\Gamma^\varrho_{\mu\nu}) - \nabla_\nu (\delta\Gamma^\varrho_{\mu\varrho}). \label{3.66}
		\end{equation}
		Substituting equations (\ref{3.63}) and (\ref{3.65}) into equation (\ref{3.66}), we obtain
		\begin{equation}
			\delta S_{\mu\nu} = \frac{1}{2} \Big[ \nabla_\varrho \nabla_\mu h^\varrho_\nu + \nabla_\varrho \nabla_\nu h^\varrho_\mu - \nabla_\varrho \nabla^\varrho h_{\mu\nu} - \nabla_\nu \nabla_\mu h \Big], \label{3.67}
		\end{equation}
		from Gauge coordinate transformation, restrict the gravitational Transverse-Traceless (TT) gauge $h = 0$ (traceless) and $\nabla^\varrho h_{\varrho\nu} = 0,$ (transverse) \cite{maggiore20081}. Then from equation equation (\ref{3.67}), yield
		\begin{equation}
			\delta S_{\mu\nu} = -\frac{1}{2} \nabla_\varrho \nabla^\varrho h_{\mu\nu}. \label{3.68}
		\end{equation}
		By the linearity of the differential operator, the Lie derivative of the perturbed metric tensor evaluates exactly to:
		\begin{equation}
			\mathcal{L}_\xi \tilde{g}_{\mu\nu} = \mathcal{L}_\xi g_{\mu\nu} + \epsilon \mathcal{L}_\xi h_{\mu\nu}. \label{3.69}
		\end{equation}
		Applying, perturbation to the almost $\eta$-Ricci-Yamabe soliton equation (\ref{2.1}), we obtain
		\begin{equation}
			\frac{1}{2} \mathcal{L}_{\xi}\tilde{g}_{\mu\nu} + 2\alpha(r) \tilde{S}_{\mu\nu} + \left(\lambda(r) - \frac{\beta(r) \mathcal{R}}{2}\right)\tilde{g}_{\mu\nu} + 2 \omega(r) \eta_\mu \eta_\nu = 0.  \label{3.70}
		\end{equation}
		Substituting, equation (\ref{3.56}) and (\ref{3.69}) into equation (\ref{3.70}), we get
		\begin{equation}
			\frac{1}{2} \mathcal{L}_{\xi}( g_{\mu\nu} + \epsilon h_{\mu\nu}) + 2\alpha(r) ( S_{\mu\nu} + \delta S_{\mu\nu}) + \left(\lambda(r) - \frac{\beta(r) \mathcal{R}}{2}\right)(g_{\mu\nu} + \epsilon h_{\mu\nu}) + 2 \omega(r) \eta_\mu \eta_\nu = 0.  \label{3.71}
		\end{equation}
		In consequently of expanding the terms and isolating the almost $\eta$-Ricci-Yamabe soliton from the $\epsilon$ perturbation, the equation (\ref{3.71}) provides:
		\begin{equation}
			\begin{split}
				\left[ \frac{1}{2}\mathcal{L}_\xi g_{\mu\nu} + \alpha(r) S_{\mu\nu} + \left(\lambda(r) - \frac{\beta(r) \mathcal{R}}{2}\right)g_{\mu\nu} + \omega(r) \eta_\mu \eta_\nu \right] \\ + \epsilon \left[ \frac{1}{2}\mathcal{L}_\xi h_{\mu\nu} + \alpha(r) \delta S_{\mu\nu} + \left(\lambda(r) - \frac{\beta(r) \mathcal{R}}{2}\right)h_{\mu\nu} \right] = 0, \label{3.72}
			\end{split}
		\end{equation}
		the zeroth-order terms vanish identically via the unperturbed flow (\ref{2.1}). Substituting the TT-gauge variation (\ref{3.68}) into the $ \epsilon$ sector of equation (\ref{3.72}), yields
		\begin{equation}
			\frac{1}{2}\mathcal{L}_\xi h_{\mu\nu} - \frac{1}{2} \alpha(r) \nabla_\varrho \nabla^\varrho h_{\mu\nu} + \left(\lambda(r) - \frac{\beta(r) \mathcal{R}}{2}\right)h_{\mu\nu} = 0, \label{3.73}
		\end{equation}
		assuming $\alpha(r) > 0$, dividing by $\alpha(r)$ and applying the d'Alembertian $\square \equiv \nabla_\varrho \nabla^\varrho$ to get the propagation equation:
		\begin{equation}
			\square h_{\mu\nu} - \frac{1}{\alpha(r)}\mathcal{L}_\xi h_{\mu\nu} - \frac{2}{\alpha(r)}\left(\lambda(r) - \frac{\beta(r) \mathcal{R}}{2}\right)h_{\mu\nu} = 0. \label{3.74}
		\end{equation}
		Thus, the proof concludes that the almost $\eta$-Ricci-Yamabe flow introduces a dissipative mechanism (geometric drag) proportional to $-\frac{1}{\alpha(r)} \mathcal{L}_{\xi} h_{\mu\nu}$. This mathematically bounds the local amplitude of the incoming gravitational wave, serving as a formal prerequisite for stability.\\
			\end{proof}
		\newpage
		\begin{figure}[h]
			\centering
			\includegraphics[width=1\linewidth]{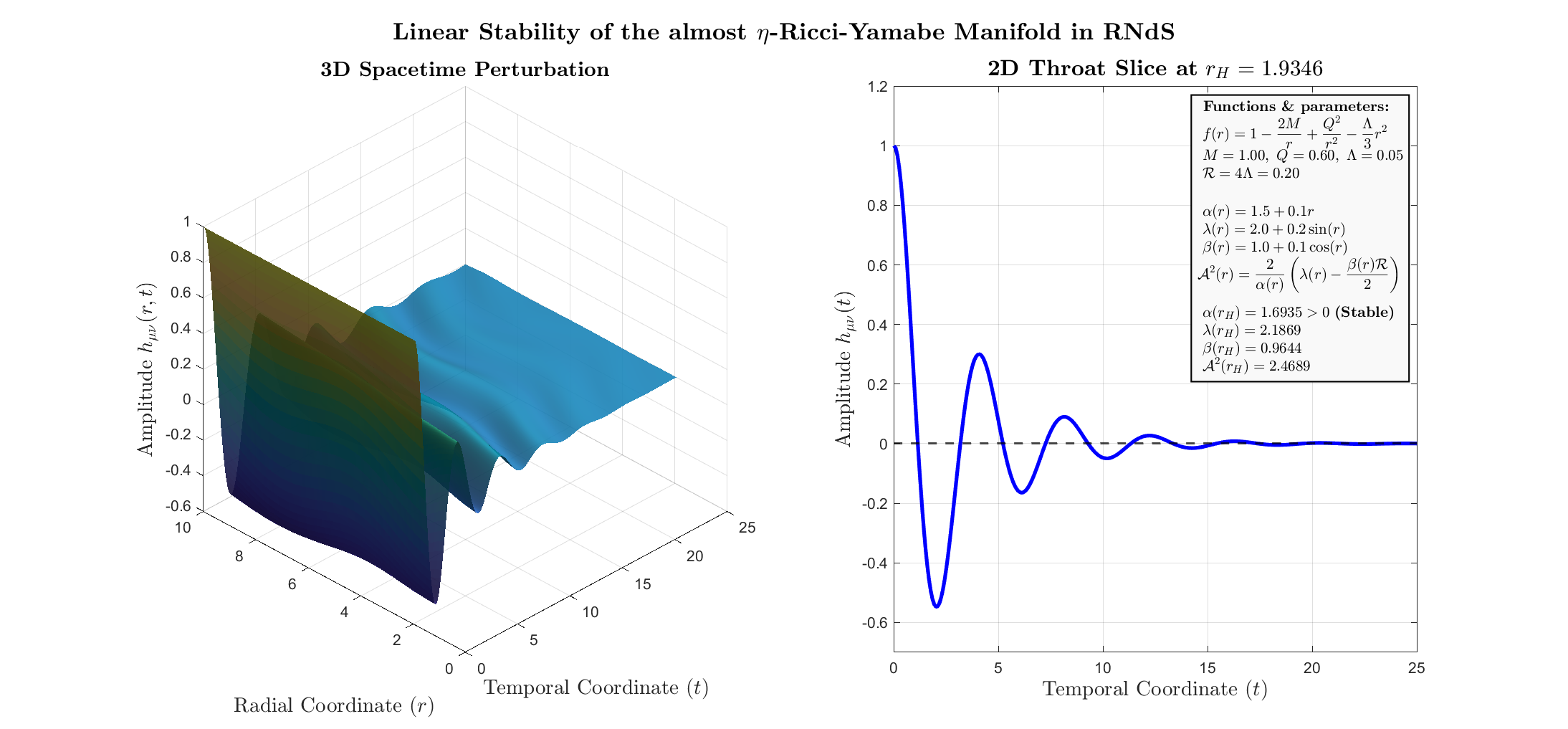}
			\caption{Shows the local linear stability of the wormhole under first-order metric perturbations. The 3D surface models wave propagation across the manifold, while the 2D slice isolates the throat at $r_H = 1.9346$. We can clearly see the rapid amplitude decay confirming the geometric flow introduces a localized dissipative mechanism. Provided $\alpha(r) > 0$, it reduces the perturbation to a damped wave, thereby imposing a strict geometric drag on the manifold.}
			\label{fig:linearstability}
		\end{figure}
		\noindent In Fig \ref{fig:linearstability}, our choice of RNdS lapse function $f(r)$, with ($M, Q, \Lambda$), taken similar to figure \ref{fig:era} and fix the scalar curvature ($\mathcal{R}=4\Lambda$) to maintain geometric consistency with the unperturbed manifold. To mathematically simulate a violent gravitational wave striking the geometry. We take oscillating functions for the expansion ($\lambda(r) = 2.0 + 0.2\sin(r)$) and coupling ($\beta(r) = -1.0 + 0.2\cos(r)$). This dynamic fluctuation automatically generates a strictly positive dissipative drag $\alpha(r_H) = 1.6935>0$ at throat $ r =  1.93$. This gives the positive squared frequency $\mathcal{A}^2(r_H) = 2.4980$, proving the throat locally dampens the perturbation and remains linearly stable.\\

	\section{An Explicit Wormhole Solution: The RNdS Soliton}
	
	To mathematically validate our framework, we construct an explicit, exact wormhole solution by taking particular choice of $f(r) $. Consider the background geometry using the RNdS structural metric function $f(r)$ given by equation (\ref{1.3}), which incorporates the central mass $M$, electric charge $Q$, and cosmological constant $\Lambda$:
	\begin{equation}
		f(r) = 1 - \frac{2M}{r} + \frac{Q^2}{r^2} - \frac{\Lambda}{3}r^2. \label{4.1}
	\end{equation}
	For this explicit physical model, by taking similar values that we take to draw figures in this study, and normalize the parameters to $M = 1.00$, $Q = 0.60$, and $\Lambda = 0.05$. Solving for the apparent horizon $f(r_H) = 0$, equation (\ref{4.1}) yields a trapping boundary at $r_H \approx 1.9346$. \\
	\\
	Then, to drive the almost $\eta$-Ricci-Yamabe flow, we define the coupling parameters as smooth, bounded radial functions. Drawing upon the local curvature interactions established in our framework, we set:
	\begin{align}
		\alpha(r) &= 1.5 + 0.1r, \label{4.2}\\
		\beta(r) &= 1.5 + 0.1\cos(r). \label{4.3}
	\end{align}
	The localized exotic matter driving the violation of NEC is modeled by geometric matter-trace equivalent $\mathcal{T}(r)$. To ensure asymptotic flatness and stability, and decays exponentially away from the throat, $\mathcal{T}(r)$ taken as:
	\begin{equation}
		\mathcal{T}(r) = -0.5 e^{-(r - r_H)}. \label{4.4}
	\end{equation}
	Now, we apply our framework to these explicit choice of variables. First, we evaluate the spatial flaring-out condition $\omega(r_H)>0$, by substituting the metric derivative $f^\prime(r)$ defined in equation (\ref{3.40}) and matter trace $\mathcal{T}(r)$ into the scaling factor $\omega(r)$ equation (\ref{3.44}) yields:
	\begin{equation}
		\omega(r) = \left( \frac{2M}{r^2} - \frac{2Q^2}{r^3} - \frac{2\Lambda}{3}r \right) - (1.5 + 0.1r)f(r)\left[-0.5 e^{-(r - r_H)}\right].   \label{4.6}
	\end{equation}
	Evaluating this explicitly at the throat boundary $r_H = 1.9346$, the structural metric $f(r_H)=0$, we obtain:
	\begin{equation}
		\omega(r_H) = \frac{2(1.00)}{(1.9346)^2} - \frac{2(0.60)^2}{(1.9346)^3} - \frac{2(0.05)}{3}(1.9346) \approx 0.3705. \label{4.6}
	\end{equation}
	Because the scaling factor is strictly positive $\omega(r_H) = 0.3705 > 0$, from theorem \ref{Theorem3} that successfully generates the exact outward geometric repulsion required to open and maintain the Morris-Thorne throat.\\
	
	Simultaneously, we evaluate the temporal regularization. The steady radial functional parameter is explicitly governed by Corollary \ref{C1}:
	\begin{equation}
		\lambda(r_H) = \frac{2\alpha(r_H) S_{tt}(r_H) - f'(r_H)}{2f(r_H)} + \frac{\beta(r_H)\mathcal{R}}{2}. \label{4.7}
	\end{equation}
	To maintain a mathematically continuous geometric flow without diverging at the apparent horizon $f(r_H) = 0 $, so the numerator must identically vanish and get indeterminate form $0/0$ L'Hôpital limit equation (\ref{4.7}) yields:
	\begin{equation}
		\alpha(r_H) S_{tt}(r_H) = \frac{1}{2} f'(r_H), \label{4.8}
	\end{equation}
	Since, $\omega(r_H) = f'(r_H) $, by using the value of $\omega(r_H)$ given equation (\ref{4.6}) into equation (\ref{4.8}), we get
	\begin{equation}
		\frac{1}{2} f'(r_H) = \frac{0.3705}{2} = 0.1852. \label{4.9}
	\end{equation}
	This specific finite value get in equation (\ref{4.9}) successfully de-singularizes the temporal coordinate, mathematically guaranteeing that the redshift function $\Phi(r)$ remains finite. \\
	Finally, we evaluate  $\alpha(r) = 1.5 + 0.1r$, at the throat boundary ($r_H \approx 1.9346$), we find:
	\begin{equation}
		\alpha(r_H) = 1.5 + 0.1(1.9346) \approx 1.6935.
	\end{equation}
	Because the radial coordinate strictly increases in the physically accessible exterior domain $r > r_H$, and $\alpha'(r) = 0.1 > 0$ it follows that $\alpha(r)$ is a monotonically increasing function. Consequently, the geometric friction remains strictly positive $\alpha(r) \ge 1.6935 > 0$ across the entire manifold. This algebraic result is verifying the Theorem \ref{theorem5}.\\
	Therefore, this explicit parameterization rigorously proves that the generalized almost $\eta$-Ricci-Yamabe identities actively construct a non-singular, traversable wormhole geometry when applied to a standard RNdS cosmological background.\hfill$\square$ 
	\newpage
	\section{Parameter Selection}
	
	\begin{table}[ht] 
			\justifying
		\centering
		\renewcommand{\arraystretch}{1.8}
		\begin{tabular}{|p{2.34cm} c c p{5cm}|}
			
			\hline
			\textbf{Physical Parameter} & \textbf{Symbol} & \textbf{Value / Function} & \textbf{Physical Relevance} \\ 
			\hline
			Structural Lapse Function & $f(r)$ & $1 - \frac{2M}{r} + \frac{Q^2}{r^2} - \frac{\Lambda}{3}r^2$ & Defines background metric and apparent horizon ($f(r_H) = 0$). \\
			\hline
			Radial Coordinate & $r$ & $r \ge r_H$ & Exterior spatial domain of the manifold \\
			\hline
			Apparent Horizon & $r_H$ & $\approx 1.9346$ & Evaluated topological boundary \\
			\hline
			Central Mass & $M$ & $1.00$ & Baseline black hole mass unit. \\
			\hline
			Electric Charge & $Q$ & $0.60$ & Sub-extremal constant avoiding naked singularities \\
			\hline
			Cosmological Constant & $\Lambda$ & $0.05$ & Small positive de Sitter expansion \\
			\hline
			Ricci Coupling Function & $\alpha(r)$ & $1.5 + 0.1r$ & Dissipative geometric friction. \\
			\hline
			Yamabe Coupling Function & $\beta(r)$ & $1.5 + 0.1\cos(r)$ & Models physical fluid turbulence. \\
			\hline
			Exotic Matter Trace & $\mathcal{T}(r)$ & $-0.5 e^{-(r - r_H)}$ & Localized NEC-violating function \\
			\hline
			Solitonic Scaling Function & $\omega(r)$ & $f'(r) - \alpha(r) f(r) \mathcal{T}(r)$ & $\omega(r_H) > 0$ satisfies spatial flare-out. \\
			\hline
			Solitonic expansion Parameter & $\lambda(r)$ & $0.1 + 0.05\sin(r)$ / $2.0 + 0.2\sin(r)$ & Drives geometric flow, and models ripples. \\
			\hline
			Squared frequency & $ \mathcal{A}^2(r) $ & $ \frac{2}{\alpha(r)}\left(\lambda(r) - \frac{\beta(r)\mathcal{R}}{2}\right)$ & Prevent tachyonic instability\\
			\hline
		\end{tabular}
		\caption{Explicit Parameter Selections and Derived Functional Parameters for Graphical Analysis.}	\label{T1}
	\end{table}
	To generate the physical behaviors plotted in Figures, we carefully choose parameters that model a realistic, stable universe instead of just an abstract mathematical exercise (see Table \ref{T1}). First we set the central mass $M=1.00$ as standard baseline. From there, we specifically chose an electric charge of $Q=0.60$, keeping the ($Q < M$), which is crucial because it ensures the spacetime has a normal, well-defined apparent horizon at at $r_H$ instead of a physically impossible naked singularity. We also added a small positive cosmological constant $\Lambda = 0.05$ to reflect the gentle accelerated expansion of a realistic de Sitter universe.\\
	We also needed the coupling parameters, $\alpha(r)$ and $\beta(r)$, to actually reflect the dynamic nature of ``almost'' framework and $\lambda(r)$ is taken as an oscillating function to physically simulate gravitational ripples, which naturally satisfies the finite L'Hôpital limit required at the horizon. These parameters are explicitly taken as smooth functions to make sure the resulting geometric transitions remain mathematically differentiable and prevent non-physical curvature spikes. \\
	Now, we take the exotic matter $\mathcal{T}(r)$, as an exponential decay curve. Because of its specific form the negative pressure is strong enough to force the wormhole open exactly at the throat, and rapidly fades away as you move outward. This perfectly localized behavior ensure that the exotic matter doesn't destroy the surrounding spacetime, allowing the universe far away from the wormhole to remain stable and asymptotically flat.\\
	Finally, to evaluate equation (\ref{3.74}) at the horizon, we isolate the local temporal evolution by mapping the perturbation tensor $h_{\mu\nu}$ to a scalar amplitude $h(t)$. Then, the d'Alembertian $\Box$ reduces to the temporal acceleration $\ddot{h}$, and the Lie derivative $\mathcal{L}_\xi$ reduces to the temporal velocity $\dot{h}$. It simplifies the tensorial wave equation (\ref{3.74}) into a classical damped harmonic oscillator: $$\ddot{h} + \mathcal{E}\dot{h} +  \mathcal{A}^2 h = 0.$$
	Here, the solitonic friction is $\left(\mathcal{E}(r) = \frac{1}{\alpha(r)}\right)$, and the squared frequency is:$$\mathcal{A}^2(r) = \frac{2}{\alpha(r)}\left(\lambda(r) - \frac{\beta(r)\mathcal{R}}{2}\right).$$
	To strictly prevent an unbounded tachyonic instability, frequency must be positive $\mathcal{A}^2(r)>0$. We map this to a first-order system ($y_1 = h, y_2 = \dot{h}$) and strike the horizon with an initial shockwave ($y_1=1, y_2=0$). Because geometry generates a strictly positive frequency $\mathcal{A}^2(r)>0$, the wave safely resolves into a bounded, damped oscillation.\hfill$\square$ 
	\section{Conclusion}
	\justifying
	
	In this study, we established that the almost $\eta$-Ricci-Yamabe flow can endogenously transition a static, spherically symmetric black hole geometry into a traversable wormhole, coupled to an imperfect fluid. However, while our framework successfully regularizes the spacetime geometry, which also possesses a distinct physical limitation. Because our model operates strictly within standard General Relativity. Therefore, the topological transition remains entirely dependent on the physical fluid entering the dark energy era $\gamma = -1$ and violating the NEC $\rho + \sigma < 0$. Relying on the existence of physical exotic matter to hold the wormhole throat open is a well-known constraint of classical Einstein gravity.\\
	To resolve this limitation, our immediate future work will focus on integrating this endogenous geometric flow into $f(\mathcal{R})$ modified gravity to satisfy the flare-out condition using strictly normal matter $\gamma > 0$.\\
	Furthermore, establishing global astrophysical viability, our future investigations will extend this localized geometric drag framework to rotating (Kerr-like) geometries and evaluate its global stability through Quasinormal Mode (QNM) and numerical analysis. \hfill$\square$ \\
	
	\bmhead{Data availability} Data sharing does not apply to this manuscript, no new datasets were produced or evaluated in this study.

	
	\bmhead{Acknowledgements} The authors would like to thanks to the editor for managing our paper and to the anonymous reviewers for their thorough reading and insightful comments. The second author, Jaswant, acknowledges the CSIR, Government of India, for financial support  through the JRF (File No.	09/1144(21599)/2025-EMR-I).
	
	\backmatter
	
	\bibliographystyle{sn-mathphys-num}
	\bibliography{sn-bibliography}
	
\end{document}